\documentclass[conference]{IEEEtran}
\IEEEoverridecommandlockouts
\usepackage[utf8]{inputenc}
\usepackage[utf8]{inputenc} 
\usepackage[T1]{fontenc}
\usepackage{url}
\usepackage{ifthen}
\usepackage{cite}
\usepackage[cmex10]{amsmath} 

\usepackage{cite}
\usepackage{amsmath,amssymb,amsfonts}
\usepackage{mathtools}
\usepackage{amsthm}
\usepackage{algorithmic}
\usepackage{graphicx}
\usepackage{textcomp}
\usepackage{tikz}
\usepackage{caption}
\usepackage{cuted}
\usepackage{romannum}
\usepackage[utf8]{inputenc}
\usepackage{pgfplots} 
\usepackage{pgfgantt}
\usepackage{pdflscape}
\usepackage{amssymb}
\usepackage{comment}
\usepackage{pst-plot}
\usetikzlibrary{spy}
\usetikzlibrary{positioning,calc}
\usetikzlibrary{decorations.pathmorphing,calc,shapes,shapes.geometric,patterns}
\usetikzlibrary{shapes.multipart}
\usepackage{xfrac}
\usepackage{colortbl}
\usepackage{cancel}
\usetikzlibrary{arrows,positioning,calc,intersections}
\usetikzlibrary{datavisualization.formats.functions}
\def\BibTeX{{\rm B\kern-.05em{\sc i\kern-.025em b}\kern-.08em
    T\kern-.1667em\lower.7ex\hbox{E}\kern-.125emX}}
    
\usepackage{pgfplots}
\usepgfplotslibrary{fillbetween}
\usetikzlibrary{arrows, decorations.markings}
\newtheorem{theorem}{Theorem}
\newtheorem{lemma}{Lemma}

\newtheorem{definition}{Definition}

\newcommand{\seta}{\ensuremath{\mathcal{A}}}

\newcommand{\setd}{\ensuremath{\mathcal{D}}}

\newcommand{\bs}[1]{\boldsymbol{#1}}

\newcommand{\RM}[1]{\MakeLowercase{\romannumeral #1{}}}

\newcommand{\mc}{\mathcal}
\newcommand{\infsub}{\mathrm{inf}}
\definecolor{calpolypomonagreen}{rgb}{0.12, 0.3, 0.17}
\newcounter{remarkcount}
\newenvironment{remark}{\refstepcounter{remarkcount}\begin{trivlist}\item \textbf{Remark \theremarkcount.}}{\end{trivlist}}
\newcommand{\circlearrow}{}
\DeclareRobustCommand{\circlearrow}{%
  \mathrel{\vphantom{\rightarrow}\mathpalette\circle@arrow\relax}%
}
\newcommand{\circle@arrow}[2]{%
  \m@th
  \ooalign{%
    \hidewidth$#1\circ\mkern1mu$\hidewidth\cr
    $#1-$\cr}%
}
\makeatother
\let\emptyset\varnothing
\usepackage{amsmath, amssymb, amsfonts, amsthm}
\usepackage{bm}
\usepackage{xcolor}
\usepackage{framed}

\usepackage{pgf}
\usepackage{tikz}
\usetikzlibrary{arrows,automata}
\usetikzlibrary{positioning}
\usetikzlibrary{decorations.pathmorphing}
\usetikzlibrary{shapes.geometric}
\usetikzlibrary{fit}
\usetikzlibrary{backgrounds}

\usepackage[caption=false,font=footnotesize]{subfig}

\newcommand{\mbf}{\mathbf}
\newcommand{\mbb}{\mathbb}

\theoremstyle{definition}

\theoremstyle{remark}

\newtheorem{corollary}{Corollary}[theorem] 

\usepackage{hyperref}
\def\BibTeX{{\rm B\kern-.05em{\sc i\kern-.025em b}\kern-.08em
    T\kern-.1667em\lower.7ex\hbox{E}\kern-.125emX}}
\begin{document}
\onecolumn

\title{Outage Common Randomness Capacity Characterization of Multiple-Antenna Slow Fading Channels\\
\thanks{H.\ Boche and M. Wiese were supported by the Deutsche Forschungsgemeinschaft (DFG, German
Research Foundation) within the Gottfried Wilhelm Leibniz Prize under Grant BO 1734/20-1, and within Germany’s Excellence Strategy EXC-2111—390814868 and EXC-2092 CASA-390781972. C.\ Deppe was supported in part by the German Federal Ministry of Education and Research (BMBF) under Grant 16KIS1005. 
H. Boche and R. Ezzine were supported by the German Federal Ministry of Education and Research (BMBF) under Grant 16KIS1003K.}
}

\author{
\IEEEauthorblockN{Rami Ezzine\IEEEauthorrefmark{1},  Moritz Wiese\IEEEauthorrefmark{1}\IEEEauthorrefmark{3}, Christian Deppe\IEEEauthorrefmark{2} and Holger Boche\IEEEauthorrefmark{1}\IEEEauthorrefmark{3}\IEEEauthorrefmark{4}}
\IEEEauthorblockA{\IEEEauthorrefmark{1}Technical University of Munich, Chair of Theoretical Information Technology, Munich, Germany\\
\IEEEauthorrefmark{2}Technical University of Munich, Institute for Communications Engineering,  Munich, Germany\\
\IEEEauthorrefmark{3}CASA -- Cyber Security in the Age of Large-Scale Adversaries–
Exzellenzcluster, Ruhr-Universit\"at Bochum, Germany\\
\IEEEauthorrefmark{4}Munich Center for Quantum Science and Technology (MCQST), Schellingstr. 4, 80799 Munich, Germany\\
Email: \{rami.ezzine, wiese, christian.deppe, boche\}@tum.de}
}
\maketitle
\thispagestyle{plain}
\pagestyle{plain}
\begin{abstract}
We investigate the problem of common randomness (CR) generation from discrete correlated sources aided by one-way communication over single-user multiple-input multiple-output (MIMO) slow fading channels with additive white Gaussian noise (AWGN), arbitrary state distribution and with channel state information available at the receiver side (CSIR). We completely solve the problem by first characterizing the channel outage capacity of MIMO slow fading channels for arbitrary state distribution. For this purpose, we also provide an achievable rate
for a specific compound MIMO Gaussian channel.
Second, we define the outage CR capacity of the MIMO slow fading channel and establish a single-letter characterization of it using our result on its outage transmission capacity.
\end{abstract}


\section{Introduction}
The availability of common randomness (CR) as a resource plays a key role in distributed computational settings\cite{survey}. It allows to design correlated random protocols that often perform faster and more efficiently than the deterministic ones.

The resource CR plays a major role in several tasks. Examples of such tasks include random coding over arbitrarily varying channels \cite{capacityAVC} and oblivious
transfer and bit commitment schemes \cite{commitmentcapacity}\cite{unconditionallysecure}. Furthermore, CR is highly relevant in the identification scheme, an approach in communications
developed by Ahlswede and Dueck \cite{identification}. It turns out that CR may allow a significant increase in the identification capacity of channels\cite{Generaltheory,part2,CRincrease}. 
In the identification framework, the decoder is not interested in knowing what the received message is. He rather wants to know if a specific message of special interest to him has been sent or not. Naturally, the sender has no knowledge of that specific message, otherwise, the problem would be trivial. While the number of identification messages (also called identities) increases exponentially with the block-length in the deterministic identification scheme, the size of the identification code increases doubly exponentially with the block-length when CR is used as a resource.
The identification scheme is more suitable than the classical transmission scheme proposed by Shannon \cite{shannon} in many practical applications which require robust and ultra-reliable low latency information exchange including several machine-to-machine
and human-to-machine systems \cite{applications}, the tactile internet \cite{tactiles}, digital watermarking \cite{Moulin,watermarkingahlswede, watermarking} and industry 4.0 \cite{industrie40}. In addition, it is worth mentioning that identification codes \cite{implementation} can be used in autonomous driving, as described in \cite{PatentBA}. Furthermore, 
CR is also of high relevance in cryptography. Indeed, under additional secrecy constraints, the generated CR can be used as secret keys, as shown in the fundamental two papers \cite{part1}\cite{Maurer}. The generated secret keys can be used to perform cryptographic tasks including secure message transmission and message authentication. In our work, however, we will not impose any secrecy requirements.

We study the problem of CR generation in the basic two-party communication setting in which
Alice and Bob aim to agree on a common random variable with high probability 
by observing independent and identically distributed (i.i.d.) samples of correlated discrete sources and while communicating as little as possible. Ahlswede and Csizár initially introduced in \cite{part2} the problem of CR generation from discrete correlated sources where the communication was over discrete noiseless channels with limited capacity. A single-letter characterization of the CR capacity for this model was established in \cite{part2}. CR capacity refers to the maximum rate of CR that Alice and Bob can generate using the resources available in the model.  The results on CR capacity were later extended to single-input single-output (SISO) and multiple-antenna Gaussian channels in \cite{CRgaussian}   for  their  practical  relevance  in many communication situations such as  wired  and  wireless communications,  satellite  and  deep  space  communication  links,  etc..
 The results on CR capacity over
Gaussian channels have been used to establish a lower-bound
on the corresponding correlation-assisted secure identification
capacity in the log-log scale in \cite{CRgaussian}. This lower bound can
already exceed the secure identification capacity over Gaussian
channels with randomized encoding elaborated in \cite{wafapaper}.

In our work, we consider the CR generation problem over MIMO slow fading channels. The focus is on the MIMO setting since multiple-antenna systems present considerable practical benefits including increased capacity, reliability and spectrum
efficiency. This is due to a combination of both diversity and spatial multiplexing gains \cite{Tse}. In particular, a practically relevant model in wireless communications is the slow fading model with  additive white Gaussian noise (AWGN)\cite{Tse,goldsmith,inftheoretic,infaspects}.
 In the multiple-antenna slow fading scenario, the channel state, represented by the channel matrix, is random but remains constant during the codeword transmission. Therefore, channel fades cannot be averaged out and ensuring reliable communication is consequently challenging.
 
An alternative commonly used concept to assess the performance in slow fading environments is the $\eta$-outage capacity defined to be the supremum of all rates for which the outage probability is lower than or equal to $\eta$\cite{Tse}\cite{goldsmith}. From the channel transmission perspective and for a given coding scheme, outage occurs when the  instantaneous channel state is so poor  that that coding scheme 
is not able to establish reliable communication over the channel. The capacity versus outage approach was initially proposed in \cite{inftheoretic} for fading channels. Later, this approach was applied to multi-antenna channels in \cite{telatar}, where the analysis was restricted to MIMO Rayleigh fading channels. However, to the best of our knowledge, no rigorous proof of the outage transmission capacity of MIMO slow fading channels with arbitrary state distribution is provided in the literature. 

The first contribution of this paper lies in establishing a single-letter  characterization of the $\eta$-outage capacity of MIMO slow fading channels with AWGN that is valid for arbitrary state distribution. To prove the capacity formula, we will additionally establish an achievable rate for a specific compound MIMO Gaussian channel.
 The second contribution of this paper lies in introducing the concept of outage in the CR generation framework as well as  characterizing the $\eta$-outage CR  capacity of MIMO slow fading channels with AWGN using our results on the corresponding $\eta$-outage transmission capacity.
In the CR generation framework, outage occurs when the channel state is so poor that Alice and Bob cannot agree on a common random variable with high probability. The $\eta$-outage CR capacity is defined to be the maximum of all achievable CR rates for which the outage probability from the CR generation perspective does not exceed $\eta.$

\quad \textit{ Paper Outline:} The rest of this paper is organized as follows. Section \ref{sec2} describes the system model and provides the key definitions as well as the auxiliary and main results. In Section \ref{proofoutagecapacity}, we provide a rigorous proof of the $\eta$-outage capacity of MIMO slow fading channels with AWGN and with arbitrary state distribution. Section \ref{proofoutagecrcapacity} is devoted to the derivation of the outage CR capacity over MIMO slow fading channels. In Section \ref{prooflowerbound}, we establish a lower bound on the capacity of a specific compound MIMO complex Gaussian channel. This result is used in the proof of the outage capacity of MIMO slow fading channels. In section \ref{application}, we study, as an application of CR generation, the problem of correlation-assisted identification over the MIMO slow fading channel and provide a lower bound on its corresponding outage correlation-assisted identification capacity. Section \ref{conclusion} contains concluding remarks and proposes
potential future research in this field.

\quad \textit{Notation:}  
$\mathbb{C}$ denotes the set of complex numbers and $\mbb R$ denotes the set of real numbers; $H(\cdot)$ and $h(\cdot)$  correspond to the entropy and the differential entropy function, respectively; $I(\cdot;\cdot)$ denotes the mutual information between two random variables. All information
quantities are taken to base 2. Throughout the paper,  $\log$ is taken to the base 2.  The natural exponential and the natural logarithm are denoted by $\exp$ and $\ln$, respectively.  For any random variables $X$, $Y$ and $Z$, we use the notation $\color{black}X \circlearrow{Y} \circlearrow{Z}\color{black}$ to indicate a Markov chain.
$|\mathcal{K}|$ stands for the cardinality of the set $\mathcal{K}$ and $\mathcal{T}_{U}^{n}$ denotes the set of typical sequences of length $n$ and of type $P_{U}$. $\text{tr}$ refers to the trace operator. For a fixed $n$-length sequence $\bs{x},$ $\mathcal{T}_{U|X}^{n}(\bs{x})$ refers to the set of sequences of length $n$ that are jointly $UX$-typical with $\bs{x}.$ For any matrix $\mbf A,$ $\lVert \mbf A\rVert$ stands for the operator norm of $\mbf A$ with respect to the Euclidean norm and  $\mbf A^{H}$ stands for the standard Hermitian transpose of $\mbf A.$ For any random matrix $\mbf A \in \mathbb{C}^{m \times n}$ with entries $\mbf A_{i,j}$ $i=1,\hdots,m,j=1,\hdots, n,$ we define
\[
\mbb E \left[\mbf A\right] = \begin{bmatrix} 
    \mbb E\left[\mbf A_{11}\right] &  \mbb E\left[\mbf A_{12}\right] & \dots \\
    \vdots & \ddots & \\
    \mbb E\left[\mbf A_{m1}\right] &        &  \mbb E\left[\mbf A_{mn}\right]
    \end{bmatrix}.
\]

$\mathcal{Q}_{P}$ is defined to be the set of positive semi-definite Hermitian matrices whose trace is smaller than or equal to $P.$ For any random vector $\bs{U},$ $\text{cov}(\bs{U})$ refers to its covariance matrix. For any random variable $X,$ \text{supp}(X) refers to its support.  For any set $\mathcal{E}$, $\mathcal{E}^c$ is its  complement. 
\section{System Model, Definitions and  Results}
\label{sec2}
\subsection{System Model for CR Generation}
\label{systemmodel}
Let a discrete memoryless multiple source $P_{XY}$ with two components, with  generic variables $X$ and $Y$ on alphabets $\mathcal{X}$ and $\mathcal{Y}$, respectively, be given.
The outputs of $X$ are observed by Terminal $A$ and those of $Y$ by Terminal $B$. Both outputs have length $n.$ Terminal $A$
can send information to Terminal $B$ over the following MIMO slow fading channel $W_{\mbf G}$:
\begin{align}
\bs{z}_{i}=\mbf G\bs{t}_{i}+\bs{\xi}_{i} \quad i=1, \hdots,n,
\label{MIMOchannelmodel}
\nonumber
\end{align}
where $\bs{t}^n=(\bs{t}_1,\hdots,\bs{t}_n)\in\mbb C^{N_{T}\times n}$ and $\bs{z}^n=(\bs{z}_1,\hdots,\bs{z}_n)\in \mbb C^{N_{R}\times n}$ are channel input and output blocks, respectively. Here, $N_T$ and $N_R$ refer to the number of transmit and receive antennas, respectively.  It is worth mentioning that the block-length $n$ can vary and that for the channel $W_\mbf G$ the arguments determine the block-length.
$\mbf G \in \mbb C^{N_{R}\times N_{T}}$ models the complex gain, where we assume that both terminals $A$ and $B$ know  the distribution of the gain $\mbf G$ and that the actual realization of the gain is known  by Terminal $B$ only.\color{black} \
$\bs{\xi}^n=(\bs{\xi}_1,\hdots,\bs{\xi}_n)\in \mathbb{C}^{N_{R}\times n}$ models the noise sequence.
We assume that the $\bs{\xi}_{i}s$ are i.i.d. such that $\bs{\xi}_{i} \sim \mathcal{N}_{\mathbb{C}}\left(0,\sigma^2 \mbf I_{N_{R}}\right), i=1, \hdots,n.$ We further assume that $\mbf G$ and $\bs{\xi}^n$ are mutually independent and that $(\mbf G,\bs{\xi}^n)$ is independent of $X^n$,$Y^n$. There are no other resources available to any of the terminals. 

A CR-generation protocol of block-length $n$ consists of:
\begin{enumerate}
    \item A function $\Phi$ that maps $X^n$ into a random variable $K$ with alphabet $\mathcal{K}$ generated by Terminal $A.$
    \item A function $\Lambda$ that maps $X^n$ into the sequence $\bs{T}^n \in \mbb C^{N_T\times n}$  satisfying the power constraint
    \begin{equation}
\frac{1}{n}\sum_{i=1}^{n}\bs{T}_{i}^H\bs{T}_{i}\leq P, \quad \text{almost surely}.   \ \
\label{energyconstraintMIMOCorrelated}
\end{equation}
    \item A function $\Psi$ that maps $Y^n$ and the  output sequence $\bs{Z}^n \in \mbb C^{N_R\times n}$ into a random variable $L$ with alphabet $\mathcal{K}$ generated by Terminal $B.$
\end{enumerate}
Such a protocol induces a pair of random variables $(K,L)$ that is called permissible.
This is illustrated in Fig. \ref{correlatedMIMO}.
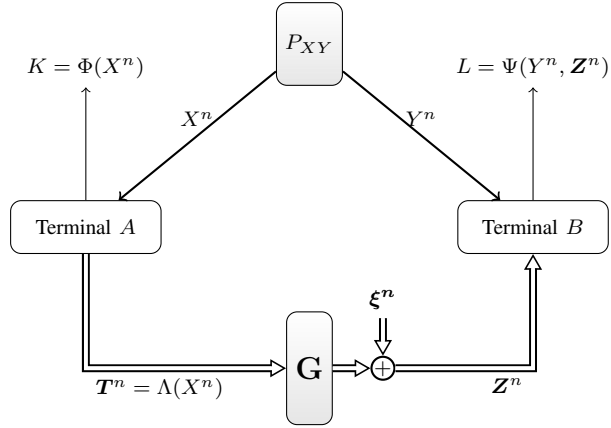
\begin{figure}[!htb]
\centering
\tikzstyle{block} = [draw, rectangle, rounded corners,
minimum height=2em, minimum width=2cm]
\tikzstyle{blocksource} = [draw, top color=white, bottom color=white!80!gray, rectangle, rounded corners,
minimum height=1.1cm, minimum width=.31cm]
\tikzstyle{blockchannel} = [draw, top color=white, bottom color=white!80!gray, rectangle, rounded corners,
minimum height=1.5cm, minimum width=.35cm]
\tikzstyle{input} = [coordinate]
\tikzstyle{vectorarrow} = [thick, decoration={markings,mark=at position
   1 with {\arrow[semithick]{open triangle 60}}},
   double distance=1.4pt, shorten >= 5.5pt,
   preaction = {decorate},
   postaction = {draw,line width=1.4pt, white,shorten >= 4.5pt}]
   \tikzstyle{sum} = [draw, circle,inner sep=0pt, minimum size=2mm,  thick]
\usetikzlibrary{arrows}
\begin{tikzpicture}[scale= 1,font=\footnotesize]
\node[blocksource] (source) {$P_{XY}$};
\node[blockchannel, below=3cm of source](channel) {\large$\mathbf{G}$};
\node[sum, right=.5cm of channel] (sum) {$+$};

\node[block, below left=2.15cm of source] (x) {Terminal $A$};
\node[block, below right=2.15cm of source] (y) {Terminal $B$};
\node[above=.5cm of sum] (noise) {$\bs{\xi^n}$};
\node[above=1.5cm of x] (k) {$K=\Phi(X^n)$};
\node[above=1.5cm of y] (l) {$L=\Psi(Y^n,\bs{Z}^n)$};

\draw[->,thick] (source) -- node[above] {$X^n$} (x);
\draw[->, thick] (source) -- node[above] {$Y^n$} (y);
\draw [vectorarrow] (x) |- node[below right] {$\bs{T}^n=\Lambda(X^n)$} (channel);
\draw [vectorarrow] (channel) -- (sum);
\draw[vectorarrow] (noise) -- (sum);
\draw[vectorarrow] (sum) -| node[below left] {$\bs{Z}^n$} (y);
\draw[->] (x) -- (k);
\draw[->] (y) -- (l);

\end{tikzpicture}
\caption{Two-correlated source model with unidirectional communication over a  MIMO slow fading channel}
\label{correlatedMIMO}
\end{figure}
\subsection{Definitions and Main Results}
We define first an achievable $\eta$-outage rate for the MIMO slow fading channel $W_{\mbf G}$ and the corresponding $\eta$-outage capacity.
For this purpose, we begin by providing the definition of a  transmission-code for  $W_{\mbf G}.$
\begin{definition}
\label{defcode}
A transmission-code $\Gamma$ of length $n$ and size $\lvert \Gamma \rvert$ for the MIMO channel $W_{\mbf G}$ is a family of pairs of codewords and decoding regions $\left\{(\mbf{t}_\ell,\setd_\ell^{(\mbf g)}),\mbf g \in \mbb C^{N_R\times N_T}, \quad \ell=1,\ldots,\lvert \Gamma \rvert \right\}$ such that for all $\ell,j \in \{1,\ldots,\lvert \Gamma \rvert\}$ and all $\mbf g \in \mbb C^{N_R \times N_T}:$ 
\begin{align}
& \mbf{t}_\ell \in \mbb C^{N_{T}\times n},\quad \setd_\ell^{(\mbf g)} \subset \mbb C^{N_{R}\times n}, \nonumber \\
&\frac{1}{n}\sum_{i=1}^{n}\bs{t}_{\ell,i}^H\bs{t}_{\ell,i}\leq P 
 \ \ \mbf{t}_\ell=(\bs{t}_{\ell,1},\hdots,\bs{t}_{\ell,n}), \nonumber \\
&\setd_\ell^{(\mbf g)}  \cap \setd_j^{(\mbf g)} = \emptyset,\quad \ell \neq j. \nonumber
\end{align}
The maximum error probability is expressed as 
\begin{align}
    e(\Gamma,\mbf g)=\underset{\ell \in \{1,\ldots,\lvert \Gamma \rvert\}}{\max}W_{\mbf g}({\setd_\ell^{(\mbf g)}}^c|\mbf{t}_\ell). \nonumber
\end{align}
\end{definition}
\begin{remark}
Throughout the paper, we consider the maximum error probability criterion.
\end{remark}
\begin{definition}
    Let $0 \leq \eta<1$. A real number $R$ is called an \textit{achievable} $\eta$-\textit{outage rate} of the channel $W_{\mbf G}$ if for every $\theta,\delta>0$ there exists a code sequence $(\Gamma_n)_{n=1}^\infty$  such that
    \[
        \frac{\log\lvert \Gamma_n\rvert}{n}\geq R-\delta
    \]
    and
    \begin{align}
        \mbb P[e(\Gamma_n,\mbf G)\leq\theta]\geq 1-\eta 
        \label{errorouterprob}
    \end{align}
    for sufficiently large $n$.
\end{definition}
\begin{remark}
The probability in \eqref{errorouterprob} is with respect to $\mbf G.$
\end{remark}
\begin{definition}
The supremum of all achievable $\eta$-outage rates is called the $\eta$-\textit{outage capacity} of the channel $W_{\mbf G}$ and is denoted by $C_\eta(P,N_{T}\times N_{R})$.
\end{definition}
\begin{theorem}
 Let $\mbf G \in \mbb C^{N_{R}\times N_{T}}$ be a random matrix. For any $\mbf Q \in \mathcal{Q}_P$ and any $R\geq 0$, let 
 \begin{align}
     \mathcal{E}(\mbf Q,R)=\Big\{\mbf g\in\mbb C^{N_R\times N_T}:\log\det(\mathbf{I}_{N_{R}}+\frac{1}{\sigma^2}\mathbf{g}\mathbf{Q}\mathbf{g}^{H})< R\Big\}.
     \label{setE}
 \end{align}
 
   The  $\eta$-outage capacity of the channel $W_{\mbf G}$ is equal to
 \begin{align}
C_{\eta}(P,N_{T}\times N_{R})=\sup \ \Big\{R: \underset{\mathbf{Q}\in\mathcal{Q}_P}{\inf }\mbb P\left[\mbf G \in \mathcal{E}(\mbf Q,R) \right] \leq \eta\Big\}. \nonumber
\end{align}
    \label{cetathmMIMO}
\end{theorem}
    \begin{corollary}
If $N_T=N_R=1,$ then
    the $\eta$-outage capacity of the SISO slow fading channel $W_{G}$ is equal to
 
    \[
        C_\eta(P)=\log\left(1+\frac{P\gamma_0^2}{\sigma^2}\right),
    \]
    where 
     $$\gamma_0=\sup\{\gamma:\mbb P[\lvert G\rvert<\gamma]\leq\eta\}.$$
    
    \end{corollary}
    
Next, we define an achievable $\eta$-outage CR rate and the $\eta$-outage CR capacity for the model presented in Section \ref{systemmodel}. This is an extension of the definition of an achievable CR rate and of the CR capacity over rate-limited discrete noiseless channels introduced in \cite{part2}.
\begin{definition} 
Fix a non-negative constant $\eta<1.$ A number $H$ is called an achievable $\eta$-outage CR rate  if there exists a non-negative constant $c$ such that for every $\alpha>0$ and $\delta>0$ and for sufficiently large $n$ there exists a permissible   pair of random variables $(K,L)$ such that
\begin{equation}
    \mbb P\left[\mbb P\left[K\neq L|\mbf G\right]\leq \alpha \right]\geq 1-\eta, 
    \label{errorMIMOcorrelated}
\end{equation}
\begin{equation}
    |\mathcal{K}|\leq 2^{cn},
    \label{cardinalityMIMOcorrelated}
\end{equation}
\begin{equation}
    \frac{1}{n}H(K)> H-\delta.
     \label{rateMIMOcorrelated}
\end{equation}
\end{definition}
\begin{remark}
The constant $\alpha>0$ in \eqref{errorMIMOcorrelated} refers to the maximum error probability from the common randomness generation perspective.
\end{remark}
\begin{remark}
The outer probability in \eqref{errorMIMOcorrelated} is with respect to $\mbf G.$
\end{remark}
\begin{remark}
Together with \eqref{errorMIMOcorrelated}, the technical condition  \eqref{cardinalityMIMOcorrelated} ensures for every $\epsilon>0$ and sufficiently large block-length $n$ that $\mbb P\left[\mbf G \in \seta^{(n,\epsilon)}\right]\geq 1-\eta,$
where
$$\seta^{(n,\epsilon)} = \bigg\{ \mbf g \in \mbb C^{N_{R}\times N_{T}}: \bigg| \frac{H(K|\mbf G=\mbf g)}{n}-\frac{H(L|\mbf G=\mbf g)}{n} \bigg| \leq \epsilon \bigg\}.$$
 This follows from the analogous statement in\cite{part2}.
\end{remark}
\begin{definition} 
The $\eta$-outage CR capacity $C_{\eta,CR}(P,N_{T}\times N_{R})$ is the maximum achievable $\eta$-outage CR rate.
\end{definition}
\begin{theorem}
For the model described in Section \ref{systemmodel}, the $\eta$-outage CR capacity is equal to
\begin{align}
C_{\eta,CR}(P,N_{T}\times N_{R}) = 
  \underset{ \substack{U \\{\substack{U \circlearrow{X} \circlearrow{Y}\\ I(U;X)-I(U;Y) \leq C_{\eta}(P,N_{T}\times N_{R})}}}}{\max} I(U;X).  \nonumber
\end{align}
\label{ccretathmMIMO}
\end{theorem}
\subsection{Auxiliary Result}
For the proof of Theorem \ref{cetathmMIMO}, we require the following result about an achievable rate for a specific compound MIMO complex Gaussian channels. Let $a>0$ be fixed arbitrarily. We consider the set $\mathcal{G}_a$ defined as 
 \begin{equation}
    \mathcal{G}_a=\{ \mbf g \in \mbb C^{N_{R}\times N_{T}}: \lVert \mbf g \rVert \leq a\}.
    \label{setga}
\end{equation}
Let $\mathcal{G}$ be any closed subset of $\mathcal{G}_a.$ We define the compound channel $$\mathcal{C}=\{ W_{\mbf g}: \mbf g \in \mathcal{G}\}.$$ 
We define next an achievable transmission rate and the transmission capacity for the compound channel $\mc C.$
\begin{definition}
 A real number $R$ is called an \textit{achievable} rate for the compound channel $\{ W_{\mbf g}: \mbf g \in \mc G\} $ if for every $\theta,\delta>0$ and all $\mbf g \in \mc G$ there exists a code sequence $(\Gamma_n)_{n=1}^\infty$  such that
    \[
        \frac{\log\lvert \Gamma_n\rvert}{n}\geq R-\delta
    \]
    and 
    \[
      \ e(\Gamma_n,\mbf g)\leq\theta,
    \]
    for sufficiently large $n$.
\end{definition}
\begin{theorem}
\label{achievratecompoundchannels}
 An achievable rate for  $\mathcal{C}$ is  $$\underset{\mbf Q \in \mathcal{Q}_{P}}{\sup}\underset{\mbf g \in \mathcal{G}}{\inf} \log\det(\mbf I_{N_{R}}+\frac{1}{\sigma^2}\mbf g \mbf Q \mbf g^H).$$
\end{theorem}

\section{Proof of Theorem \ref{cetathmMIMO}}
\label{proofoutagecapacity}
\subsection{Direct Proof}
 We define
    \begin{align}
       R_{\eta,\sup} 
       &=\sup \ \Big\{R: \underset{\mathbf{Q}\in \mathcal{Q}_{P}}{\inf }\mbb P\left[\mbf G \in \mathcal{E}(\mbf Q,R) \right] \leq \eta\Big\}.\nonumber
    \end{align}
We will show that
    \begin{equation}\label{eq:outage_geq}
        C_\eta(P,N_{T}\times N_{R})\geq R_{\eta,\sup}-\epsilon, \nonumber
    \end{equation}
    with $\epsilon$ being an arbitrarily small positive constant.
    Clearly, from the definition of $R_{\eta,\sup}$ above, it holds that
    \begin{align}
         P_\infsub=\underset{\mbf Q \in \mathcal{Q}_{P}}{\inf }\mbb P\left[\mbf G \in \mathcal{E}(\mbf Q,R_{\eta,\sup}-\frac{\epsilon}{2}) \right] \leq \eta.
        \nonumber
    \end{align}
    Next, we choose a $\hat{\mbf Q} \in \mathcal{Q}_P $ such that $ \mbb P \left[\mbf G \in \mathcal{E}(\hat{\mbf Q},R_{\eta,\sup}-\frac{\epsilon}{2}) \right] \leq \eta.$

    Here, we distinguish two cases:
    \paragraph{If  \texorpdfstring{$P_\infsub< \eta$}{TEXT}}We fix $\alpha_1>0$ to be sufficiently small such that $P_\infsub+\alpha_1\leq \eta.$
    We choose a 
    $\hat{\mbf Q}\in \mathcal{Q}_P$ such that
    \begin{align}
        \mbb P\left[\mbf G \in \mathcal{E}(\hat{\mbf Q},R_{\eta,\sup}-\frac{\epsilon}{2}) \right]
        &\leq P_\infsub+\alpha_1 \nonumber \\
        &\leq \eta \nonumber \nonumber\\
        &<1.
        \nonumber
    \end{align}
    \paragraph{\text{If} \texorpdfstring{$  P_\infsub=\eta$}{TEXT}} We choose a $\hat{\mbf Q}\in\mathcal{Q}_P$ such that $\mbb P\left[\mbf G \in \mathcal{E}(\hat{\mbf Q},R_{\eta,\sup}-\frac{\epsilon}{2}) \right]=P_\infsub=\eta<1.$ The existence of such a $\hat{\mbf Q}$ is proved in what follows:
    \begin{lemma}
    If $P_\infsub=\eta,$ then there exists a $\hat{\mbf Q} \in \mathcal{Q}_P$ such that
    \begin{align}
    P_\infsub=\mbb P\left[\mbf G \in \mathcal{E}(\hat{\mbf Q},R_{\eta,\sup}-\frac{\epsilon}{2})  \right]. \nonumber 
    \end{align}
    \end{lemma}
    \begin{proof}
	Suppose that the claim of the lemma is not true. Then $P_\infsub=\eta$, but 
	\begin{align}
		\mbb P[\mbf G \in \mc E(\mbf Q,R_{\eta,\sup}-\frac{\epsilon}{2})]>\eta
		\label{retasupminusepsilon}
	\end{align}
	for all $\mbf Q\in\mc Q_P$. Set
	\begin{align}
		R(\mbf Q)=\sup\{R:\mbb P[\mbf G \in \mc E(\mbf Q,R)]\leq\eta\}.
		\label{RQ}
	\end{align}
It follows from the well-known Lemma \ref{sigmacontinuity} and from using \eqref{retasupminusepsilon} that $R(\mbf Q)<R_{\eta,\sup}-\frac{\epsilon}{2}$ for all $\mbf Q\in\mc Q_P$.
	
	\textbf{Claim:} For every $\mbf Q\in\mc Q_P$ there exists $\eta_{\mbf Q}>\eta$ and $\delta_{\mbf Q}>0$ such that 
	\[
		\lVert\mbf Q'-\mbf Q\rVert<\delta_{\mbf Q}\Longrightarrow\mbb P[\mbf G \in \mc E(\mbf Q',R_{\eta,\sup}-\frac{\epsilon}{2})]>\eta_{\mbf Q}.
	\]
	Assume the claim is true. Then by the compactness of $\mc Q_P$, there exist $\mbf Q_1,\ldots,\mbf Q_l$ such that 
	\[
		\mc Q_P=\bigcup_{i=1}^l\{\mbf Q \in \mc Q_{P}:\lVert\mbf Q-\mbf Q_i\rVert<\delta_{\mbf Q_i}\}.
	\]
	It follows that for every $\mbf Q\in\mc Q_P$, we have
	\[
		\mbb P[\mbf G \in \mc E(\mbf Q,R_{\eta,\sup}-\frac{\epsilon}{2})]
		\geq\min_{i=1,\ldots,l}\eta_{\mbf Q_i}
		>\eta,
	\]
	which implies that
	\begin{align}
    \underset{\mbf Q \in \mc Q_P}{\inf}\mbb P[\mbf G \in \mc E(\mbf Q,R_{\eta,\sup}-\frac{\epsilon}{2})]>\eta.
	\nonumber \end{align}
	This is a contradiction to the assumption of the lemma that $P_\infsub=\eta$. Therefore, there must exist a $\hat{\mbf Q}\in \mc Q_P$ satisfying $P_\infsub=\mbb P[\mbf G \in \mc E(\hat{\mbf Q},R_{\eta,\sup}-\frac{\epsilon}{2})]$. 
	
	Now it remains to prove the claim. Let $\mbf Q \in \mc Q_{P}.$ Let $f(\mbf g,\mbf Q)=\log\det(\mbf I_{N_{R}}+\frac{1}{\sigma^2}\mbf g \mbf Q \mbf g^{H}).$
	We define
	\[
		\epsilon'=\frac{R_{\eta,\sup}-\frac{\epsilon}{2}-R(\mbf Q)}{2}.
	\]
	Then $\epsilon'>0$, and so
	\[
		\mbb P[f(\mbf G,\mbf Q)<R(\mbf Q)+\epsilon']>\eta.
	\]

	Choose $a$ so large that 
	\begin{equation}\label{eq:a_choice}
		\mbb P[\lVert \mbf G\rVert>a]<\frac{\mbb P[f(\mbf G,\mbf Q)<R(\mbf Q)+\epsilon']-\eta}{2}.
	\end{equation}
    By Lemma \ref{probzerolemma}, we know that such an $a$ exists. Now, notice that
 the function $f$ is uniformly continuous on $\mc G_a\times\mc Q_P$. We recall that $\mc G_a=\{\mbf g \in \mbb C^{N_R\times N_T}:\lVert\mbf g\rVert\leq a\}.$ Equip the set $\mc G_a\times\mc Q_P$ with the norm 
	\[
		\lVert(\mbf g,\mbf Q)\rVert=\lVert\mbf g\rVert+\lVert\mbf Q\rVert.
	\]
	It follows that there exists a $\delta_{\mbf Q}>0$ such that for any $(\mbf g_1,\mbf Q_1),(\mbf g_2,\mbf Q_2)\in\mc G_a\times\mc Q_P$,
	\[
		\lVert(\mbf g_1,\mbf Q_1)-(\mbf g_2,\mbf Q_2)\rVert<\delta_{\mbf Q}\quad\Longrightarrow\quad\lvert f(\mbf g_1,\mbf Q_1)-f(\mbf g_2,\mbf Q_2)\rvert<\epsilon'.
	\]
	
	Now take any $\mbf Q'\in\mc Q_P$ satisfying $\lVert\mbf Q'-\mbf Q\rVert<\delta_{\mbf Q}$. By the above, $\lvert f(\mbf g,\mbf Q')-f(\mbf g,\mbf Q)\rvert<\epsilon'$ for any $\mbf g\in\mc G_a$. Hence, since $R_{\eta,\sup}-\frac{\epsilon}{2}-\epsilon'=R(\mbf Q)+\epsilon'$, it follows that
	\begin{align*}
		\mbb P[f(\mbf G,\mbf Q')<R_{\eta,\sup}-\frac{\epsilon}{2}]
		&\geq\mbb P[f(\mbf G,\mbf Q')<R_{\eta,\sup}-\frac{\epsilon}{2},\mbf G\in\mc G_a]\\
		&\geq\mbb P[f(\mbf G,\mbf Q)<R_{\eta,\sup}-\frac{\epsilon}{2}-\epsilon',\mbf G\in\mc G_a]\\
		&\geq\mbb P[f(\mbf G,\mbf Q)<R(\mbf Q)+\epsilon']-\mbb P[\mbf G\notin\mc G_a]\\
		&=:\eta_{\mbf Q}.
	\end{align*}
By the choice of $a$ (see \eqref{eq:a_choice}), we have $\eta_{\mbf Q}>\eta$. This proves the claim and thus completes the proof of the lemma.
    \end{proof}
    	\begin{lemma}
	$\forall \mbf Q \in \mc Q_{P}:$
	$$\mbb P \left[ \mbf G \in \mc E(\mbf Q, R(\mbf Q))  \right] \leq \eta$$ and so the supremum in \eqref{RQ}  is actually a maximum.
	\label{sigmacontinuity}
	\end{lemma}

	\begin{proof}
    Let $R_n\nearrow R(\mbf Q)$ be a sequence converging to $R(\mbf Q)$ from the left. Then 
    \[
        \{R\in \mbb R:R <R(\mbf Q)\}=\bigcup_{n=1}^\infty\{R\in \mbb R:R <R_n\}.
    \]
    From the sigma-continuity of probability measures, it follows that 
        $$\mbb P[ \mbf G \in \mc E(\mbf Q, R(\mbf Q))]=\underset{n\rightarrow \infty}{\lim}\mbb P[\mbf G\in \mc E(\mbf Q,R_n)]\leq\eta.$$
\end{proof}
	\begin{lemma}
	\label{probzerolemma}
	$$ \underset{a\rightarrow \infty}{\lim} \mbb P \left[ \lVert \mbf G \rVert > a \right]=0$$
	\end{lemma}
	\begin{proof}
		To prove the lemma, it suffices to notice that
 $$\bigcap_{n=1}^{\infty} \{ \lVert \mbf G \rVert >n \}=\varnothing.$$ It follows that $$\underset{m\rightarrow \infty}{\lim}\mbb P\left[\bigcap_{n=1}^{m} \{ \lVert \mbf G \rVert >n \} \right]=0. $$
 As a result:
 \begin{align}
     \underset{a\rightarrow \infty}{\lim} \mbb P \left[ \lVert \mbf G \rVert > a \right]=0.\nonumber
 \end{align}
	\end{proof}

So far, we have proved the existence of a $\hat{\mbf Q} \in \mc Q_P$ such that
$$ \mbb P[\mbf G \in \mc E(\hat{\mbf Q},R_{\eta,\sup}-\frac{\epsilon}{2})]\leq \eta.$$
  For some $a>0,$ we consider now the set
   \begin{align}
       \hat{\mc G}=\Big\{ \mbf g \in \mbb C^{N_R\times N_T}:R_{\eta,\sup}-\frac{\epsilon}{2}\leq \log\det(\mbf I_{N_{R}}+\frac{1}{\sigma^2}\mbf g \hat{\mbf Q} \mbf g^{H}) \ \text{and} \ \lVert \mbf g \rVert \leq a    \Big\}.
   \nonumber \end{align}
   We choose a non-singular $\tilde{\mbf Q} \in \mc Q_P$ such that for $\alpha_2>0$ sufficiently small
   	\begin{align}
   	    	\lVert \hat{\mbf Q}-\tilde{\mbf Q} \rVert \leq \alpha_2 \quad \Longrightarrow \quad  
     	\lvert f(\mbf g,\hat{\mbf Q})-f(\mbf g,\tilde{\mbf Q})\rvert \leq \frac{\epsilon}{2} \quad \forall \mbf g \in \hat{\mc G}.     	\label{choiceofhatQ}
   	\end{align}
	From the uniform continuity of $f$ on the compact set $\mc G_a \times \mc Q_P,$ we know that such a $\tilde{\mbf Q} \in \mc Q_P$ exists.
	Now, consider the set
   $$\tilde{\mathcal{G}}=\{ \mbf g \in \mbb C^{N_{R}\times N_{T}}: R_{\eta,\sup}-\epsilon\leq \log\det(\mbf I_{N_{R}}+\frac{1}{\sigma^2} \mbf g  \tilde{\mbf Q} \mbf g^{H}) \ \text{and} \ \lVert \mbf g \rVert \leq a \}.   $$
   \begin{lemma}
   \label{hatgincludedintildeg}
   \begin{align}
         \hat{\mc G} \subseteq \tilde{\mc G} \nonumber
   \end{align}
   \end{lemma}
   \begin{proof}
    Let $\mbf g \in \hat{\mc G}$ be fixed arbitrarily.  It holds that
   \begin{align}
       \lVert \mbf g \rVert \leq a 
       \label{firstreq}
   \end{align}
   and that
   \begin{align}
       R_{\eta,\sup}-\frac{\epsilon}{2}\leq \log\det(\mbf I_{N_{R}}+\frac{1}{\sigma^2}\mbf g \hat{\mbf Q} \mbf g^{H}).\nonumber
   \end{align}
   From \eqref{choiceofhatQ}, it follows that
   \begin{align}
       \log\det(\mbf I_{N_{R}}+\frac{1}{\sigma^2}\mbf g \hat{\mbf Q} \mbf g) \leq \frac{\epsilon}{2}+  \log\det(\mbf I_{N_{R}}+\frac{1}{\sigma^2}\mbf g \tilde{\mbf Q} \mbf g),\nonumber
   \end{align}
   yielding
   \begin{align}
   \log\det(\mbf I_{N_{R}}+\frac{1}{\sigma^2}\mbf g \tilde{\mbf Q} \mbf g) \geq R_{\eta,\sup}-\epsilon. \label{secondreq}
   \end{align}
  The inequalities \eqref{firstreq} and \eqref{secondreq} imply that
  $\mbf g \in \tilde{\mc G}.$
   \end{proof}
   Now, since the set $\Big\{ \mbf g \in \mbb C^{N_R\times N_T}:R_{\eta,\sup}-\frac{\epsilon}{2}\leq \log\det(\mbf I_{N_{R}}+\frac{1}{\sigma^2}\mbf g \hat{\mbf Q} \mbf g^{H})\Big\}$ is closed, it follows that $\tilde{\mc G}$ is a closed subset of $\mc G_a.$ By applying Theorem \ref{achievratecompoundchannels}, it follows that 
\begin{align}
    \underset{\mbf Q \in \mathcal{Q}_{P}}{\sup}\underset{\mbf g\in\tilde{\mathcal{G}}}{\inf} \log\det(\mathbf{I}_{N_{R}}+\frac{1}{\sigma^2}\mathbf{g}\mathbf{Q}\mathbf{g}^{H}) \nonumber
\end{align} \color{black}
is an achievable rate for $\tilde{\mc C}=\{W_{\mbf g}: \mbf g\in \tilde{\mathcal{G}}\}.$ 
Since $\tilde{\mbf Q}\in \mathcal{Q}_P$, it follows
that $$\underset{g\in\tilde{\mathcal{G}}}{\inf} \log\det(\mathbf{I}_{N_{R}}+\frac{1}{\sigma^2}\mathbf{g}\tilde{\mathbf{Q}}\mathbf{g}^{H})$$ is also an achievable rate for $\tilde{\mc C}$.
   
  Let $\theta,\delta>0$.
Since $\color{black}R_{\eta,\sup}-\epsilon \leq \underset{g\in\tilde{\mathcal{G}}}{\inf} \log\det(\mathbf{I}_{N_{R}}+\frac{1}{\sigma^2}\mathbf{g}\tilde{\mathbf{Q}}\mathbf{g}^{H})\ \color{black}$, it exists a code sequence $(\Gamma_{\tilde{\mc G},n})_{n=1}^\infty$ and a block length $n_0$ such that
    \[
        \frac{\log\lvert \Gamma_{\tilde{\mc G},n}\rvert}{n}\geq R_{\eta,\sup}-\epsilon-\delta
    \]
    and such that 
 \begin{align}
    &\mbf g \in \tilde{\mathcal{G}} \implies e(\Gamma_{\tilde{\mc G},n},\mbf g) \leq \theta     \nonumber 
    \end{align}
    for $n\geq n_0.$
    By Lemma \ref{hatgincludedintildeg}, it follows for $n\geq n_0$ that
    $$ e(\Gamma_n,\mbf g) \leq \theta \quad \forall \mbf g \in \hat{\mc G}.$$
    Next, we will prove that for a suitable $a>0$ and for $n\geq n_0$
    \begin{align}
     \mbf g \in \mc G_a^{c}=\{\mbf g \in \mbb C^{N_R\times N_T}: \lVert \mbf g \rVert >a    \} \Longrightarrow e(\Gamma_{\mbf g,n},\mbf g) \leq \theta, \nonumber
    \end{align}
    where $\Gamma_{\mbf g,n}$ is some code  with block-length $n.$
    Since $\tilde{\mbf Q}\in \mc Q_P$ is non-singular, it holds that
    \begin{align}
       \underset{a\rightarrow \infty}{\lim} \underset{\lVert \mbf g \rVert=a}{\min}f(\mbf g,\tilde{ \mbf Q})=\infty. \nonumber
    \end{align}
 One can therefore choose $a>0$ sufficiently large such that 
 \begin{align}
 \Big\{\mbf g \in \mbb C^{N_R\times N_T}: \lVert \mbf g
 \rVert=a\Big\}\subseteq \tilde{\mc G}.
 \label{reqa}
 \end{align}
 We fix $a>0$ satisfying \eqref{reqa}.
 Let $\mbf g$ with $\lVert \mbf g \rVert >a$ be fixed arbitrarily. Then, there exists  $\mbf g'=\frac{a}{\lVert \mbf g \rVert}\mbf g\in \tilde{\mc G}$ such that the Gaussian channel $W_{\mbf g'}$ is degraded from the Gaussian channel $W_\mbf g.$ It follows that there exists a code sequence $(\Gamma_{\mbf g,n})_{n=1}^{\infty}$ for $W_\mbf g$ such that  each $\Gamma_{\mbf g,n}$ has the same encoder and the same size as $\Gamma_{\tilde{\mc G},n}$ but a different decoder adjusted to $\mbf g$ and such that for $n\geq n_0,$ $e(\Gamma_{\mbf g,n},\mbf g)\leq \theta$. Here, we require channel state information at the receiver side (CSIR) so that the decoder can adjust its decoding strategy according to the channel state.
 \begin{remark}
In the above we used the fact that if for a degraded Gaussian channel $W_{\mbf g'}$, a code $\Gamma'$ satisfies $e(\Gamma',\mbf g') \leq \theta$, then there exists a code $\Gamma$ for $W_{\mbf g}$, the channel from which it was degraded with the same encoder as $\Gamma'$ but with possibly a different decoder such that  $e(\Gamma,\mbf g)\leq \theta$ . The analogous statement is provided in Problem 6.16 in \cite{codingtheorems} for discrete memoryless channels (DMCs). This is a special case of the statement provided in \cite{NoteShannon} for DMCs.
 \end{remark}
 
So far, we have proved the existence of a
code sequence $(\Gamma_{n})_{n=1}^\infty$ and a block length $n_0$ such that
    \[
        \frac{\log\lvert\Gamma_{n}\rvert}{n}\geq R_{\eta,\sup}-\epsilon-\delta
    \]
    and such that 
 \begin{align}
    &\mbf g \in \hat{\mathcal{G}}\cup \mc G_a^{c} \implies e(\Gamma_{n},\mbf g) \leq \theta     \nonumber 
    \end{align}
    for $n\geq n_0.$
    Now, we have for $n\geq n_0$
    \begin{align}
       \mbb P \left[e(\Gamma_{n},\mbf G) \leq \theta\right]& \geq \mbb P \left[ \mbf G \in \hat{\mathcal{G}}\cup \mc G_a^{c}  \right] \nonumber\\
       &=\mbb P\left[ \mbf G \in \hat{\mathcal{G}}\right]+\mbb P \left[ \mbf G \in \mc G_a^{c}   \right] \nonumber \\
       &\geq \mbb P \left[ \mbf G \notin \mc E(\hat{\mbf Q},R_{\eta,\sup}-\frac{\epsilon}{2})  \right] \nonumber \\
       &\geq 1-\eta. \nonumber
    \end{align}
    This completes the direct proof of Theorem \ref{cetathmMIMO}.
\subsection{Converse Proof}

We will prove now the converse, i.e, we will show that
\begin{align}
    C_{\eta}(P,N_{T}\times N_{R})\leq R_{\eta,\sup}.
    \label{conversequation}
\end{align}
\begin{remark}
It is here worth mentioning that the weak converse for compound channels does not guarantee that the error probability cannot be made arbitrarily small for all possible states when the target rate exceeds the compound capacity. Therefore, we cannot use the weak converse theorem of compound channels to prove the converse of Theorem \ref{cetathmMIMO}. We will proceed differently. 
\end{remark}
   Suppose \eqref{conversequation} were not true. Then there exists an $\epsilon>0$ such that $R_{\eta,\sup}+\epsilon$ is an achievable $\eta$-outage rate for $W_{\mbf G}$. The goal is to find a contradiction. Choose $\theta>0$ so small that 
	\[
		\frac{(1-\theta)\epsilon}{2}-\theta R_{\eta,\sup}>\frac{\epsilon}{4}.
	\]
	Due to the achievability of $R_{\eta,\sup}+\epsilon,$ there exists a code sequence $(\Gamma_n)_{n=1}^{\infty}$ such that 
	\begin{align}
		\frac{\log\lvert\Gamma_n\rvert}{n}&\geq R_{\eta,\sup}+\frac{\epsilon}{2} 
		\label{ratecondition}
		\end{align}
		and
		\begin{align}
		\mbb P[e(\Gamma_n,\mbf G)>\theta]&\leq\eta
		\label{erro2}
	\end{align}
	for sufficiently large $n.$ Choose an $n$ for which the above holds and which satisfies
	\begin{align}
		\frac{1}{n}\leq\frac{\epsilon}{8}.
		\label{fixedblocklength}
	\end{align}

We fix the covariance matrices $\mbf Q_1,\dots, \mbf Q_n$ of the random channel inputs $\bs{T}_1,\hdots,\bs{T}_n,$ respectively and let $\mbf  Q^{\star}=\frac{1}{n}\sum_{i=1}^{n}\mbf Q_{i}.$ 
	Furthermore, we let 
	\[
		\epsilon'=\frac{\epsilon}{8}.
	\]
We consider the following two sets:
\begin{align}
    \mathcal{E}(\mbf Q^\star,R_{\eta,\sup}+\epsilon')=\{\mbf g \in \mbb C^{N_R\times N_T}: \log\det(\mbf I_{N_{R}}+\frac{1}{\sigma^2}\mbf g \mbf Q^{\star} \mbf g^{H})< R_{\eta,\sup}+\epsilon'    \}
 \nonumber
\end{align}
and
\begin{align}
    \mathcal{G}_\theta
    &=\{\mbf g \in \mathcal{E}(\mbf Q^\star,R_{\eta,\sup}+\epsilon') : e(\Gamma_n,\mbf g)\leq \theta     \}.
\nonumber
\end{align}
The goal is to prove that the set $\mc G_{\theta}$ is non-empty. For this purpose, we will show that $\mbb P \left[\mbf G \in  \mathcal{E}(\mbf Q^\star,R_{\eta,\sup}+\epsilon')\right] > \eta$ in what follows:
\begin{lemma}
\label{probbiggereta}
\begin{align}
    \mbb P\left[\mbf G \in \mathcal{E}(\mbf Q^\star,R_{\eta,\sup}+\epsilon')\right]>\eta. \nonumber
    \end{align}
\end{lemma}
\begin{proof}
By Lemma \ref{traceQ}, we know that $\text{tr}(\mbf Q^\star)\leq P$ and therefore $\mbf Q^\star \in \mathcal{Q}_P.$ By Lemma \ref{supQ}, it follows that
\begin{align}
    R(\mbf Q^\star)&= \sup \Big\{R: \mbb P\left[\mbf G \in \mathcal{E}(\mbf Q^\star,R)\right]\leq \eta \Big\} \nonumber \\
    &\leq R_{\eta,\sup}. \nonumber
\end{align}
This yields
\begin{align}
    \mbb P \left[ \log\det(\mbf I_{N_{R}}+\frac{1}{\sigma^2}\mbf G \mbf Q^\star \mbf G^{H})< R_{\eta,\sup}+\epsilon' \right] 
    &\geq \mbb P \left[ \log\det(\mbf I_{N_{R}}+\frac{1}{\sigma^2}\mbf G \mbf Q^\star \mbf G^{H})< R(\mbf Q^\star)+\epsilon' \right] \nonumber \\
    &>\eta. \nonumber
\end{align}
\end{proof}
\begin{lemma}
\begin{align}
    \mathrm{tr}(\mbf Q^\star)\leq P. \nonumber
\end{align}
\label{traceQ}
\end{lemma}
\begin{proof}
From \eqref{energyconstraintMIMOCorrelated}, it holds that
\begin{align}
    \frac{1}{n}\sum_{i=1}^{n} \bs{T}_{i}^H\bs{T}_{i}\leq P, \quad \text{almost surely}. \nonumber
\end{align}
This implies that
\begin{align}
    \mbb E\left[\frac{1}{n}\sum_{i=1}^{n} \bs{T}_{i}^H\bs{T}_{i}\right] &=\frac{1}{n}\sum_{i=1}^{n}\mbb E\left[ \bs{T}_{i}^H\bs{T}_{i}\right] \nonumber \\
    &\leq P.\nonumber
\end{align}
This yields
\begin{align}
    \text{tr}\left[\mbf Q^\star\right] &= \text{tr} \left[\frac{1}{n}\sum_{i=1}^{n} \mbf Q_{i}\right] \nonumber \\
    &=\frac{1}{n}\sum_{i=1}^{n}\text{tr}\left[ \mbf Q_i\right] \nonumber\\
    &\leq \frac{1}{n}\sum_{i=1}^{n} \text{tr}\left(\mbb E\left[ \bs{T}_{i}\bs{T}_{i}^H\right]\right) \nonumber \\
    &=\frac{1}{n}\sum_{i=1}^{n} \mbb E\left[ \text{tr}\left(\bs{T}_{i}\bs{T}_{i}^H\right)\right] \nonumber \\
    &=\frac{1}{n}\sum_{i=1}^{n} \mbb E\left[ \text{tr}\left(\bs{T}_{i}^H\bs{T}_{i}\right)\right] \nonumber \\
    &=\frac{1}{n}\sum_{i=1}^{n} \mbb E\left[ \bs{T}_{i}^H\bs{T}_{i}\right] \nonumber \\
    &\leq P, \nonumber
\end{align}
where we used $r=\text{tr}(r)$ for scalar $r$, $\text{tr}\left(\mbf A \mbf B\right)= \text{tr}\left(\mbf B \mbf A\right)$ and the linearity of the expectation and of the trace operators.
\end{proof}
\begin{lemma}
\label{supQ}
For any $\mbf Q \in \mathcal{Q}_{P}$, it holds that
\begin{align}
    \sup \Big\{R: \mbb P\left[\mbf G \in \mathcal{E}(\mbf Q,R)\right]\leq \eta \Big\} 
   &\leq \sup \Big\{R: \underset{\mbf Q \in \mathcal{Q}_{P}}{\inf}\mbb P\left[\mbf G \in \mathcal{E}(\mbf Q,R)\right] \leq \eta\Big\}.\nonumber
\end{align}
\end{lemma}
\begin{proof}
For any $\mbf Q \in \mathcal{Q}_{P}$, we have
\begin{align}
    \Big\{R: \ \mbb P\left[\mbf G \in \mathcal{E}(\mbf Q,R)\right]\leq \eta \Big\} 
   &\subseteq \Big\{R:\ \underset{\mbf Q \in \mathcal{Q}_{P}}{\inf}\mbb P\left[\mbf G \in \mathcal{E}(\mbf Q,R)\right] \leq \eta\Big\} \nonumber
\end{align}
As a result:
\begin{align}
    \sup \Big\{R: \mbb P\left[\mbf G \in \mathcal{E}(\mbf Q,R)\right]\leq \eta \Big\}\leq \sup \Big\{R: \underset{\mbf Q \in \mathcal{Q}_{P}}{\inf}\mbb P\left[\mbf G \in \mathcal{E}(\mbf Q,R) \right] \leq \eta\Big\}.\nonumber
\end{align}
\end{proof}
Now, we can prove that the set $\mc G_\theta$ is non-empty. 
\begin{lemma}
 $\mathcal{G}_{\theta}$ is a non-empty set. 
\end{lemma}
\begin{proof}
By Lemma \ref{probbiggereta}, we have
\begin{align}
    \eta &<\mbb P \left[ \mbf G \in \mathcal{E}(\mbf Q^\star,R_{\eta,\sup}+\epsilon')\right] \nonumber \\
    &=\mbb P \left[ \mbf G \in\mathcal{E}(\mbf Q^\star,R_{\eta,\sup}+\epsilon')| e(\Gamma_n,\mbf G)\leq \theta \right] \mbb P \left[e(\Gamma_n,\mbf G)\leq \theta\right] +\mbb P \left[ \mbf G \in \mathcal{E}(\mbf Q^\star,R_{\eta,\sup}+\epsilon')| e(\Gamma_n,\mbf G)>\theta \right] \mbb P \left[e(\Gamma_n,\mbf G)> \theta\right]  \nonumber \\
    &\leq \mbb P \left[ \mbf G \in \mathcal{E}(\mbf Q^\star,R_{\eta,\sup}+\epsilon')| e(\Gamma_n,\mbf G)\leq \theta \right]+ \mbb P \left[e(\Gamma_n,\mbf G)> \theta\right] \nonumber \\
    &\leq \mbb P \left[ \mbf G \in \mathcal{E}(\mbf Q^\star,R_{\eta,\sup}+\epsilon')| e(\Gamma_n,\mbf G)\leq \theta \right]+\eta, \nonumber
\end{align}
where we used  \eqref{erro2} in the last step.
This implies that
\begin{align}
    \mbb P \left[ \mbf G \in \mathcal{E}(\mbf Q^\star,R_{\eta,\sup}+\epsilon')| e(\Gamma_n,\mbf G)\leq \theta \nonumber \right]>0. \nonumber
\end{align}
Furthermore, since $\eta <1$, it follows  that
\begin{align}
    \mbb P\left[ e(\Gamma_n,\mbf G)\leq \theta\right]\geq 1-\eta >0. \nonumber
\end{align}
As a result, we have
\begin{align}
       \mbb P \left[ \mbf G \in \mathcal{E}(\mbf Q^\star,R_{\eta,\sup}+\epsilon'), e(\Gamma_n,\mbf G)\leq \theta \nonumber \right]>0, \nonumber
\end{align}
which means that
\begin{align}
\mbb P \left[\mbf G \in \mathcal{G}_\theta\right]>0 \nonumber
\end{align}
and therefore $\mc G_\theta$ is a non-empty set.
\end{proof}
Pick a $\mbf g \in \mathcal{G}_{\theta}$ and consider the channel
\begin{align}
\bs{z}_{i}=\mbf g\bs{t}_{i}+\bs{\xi}_{i} \quad i=1, \hdots,n.  \nonumber
\end{align}

We model the random input and output sequence by $\bs{T}^n$ and $\bs{Z}^n,$ respectively. Furthermore, the random message is modeled by $W$ and the random decoded message is modeled by $\hat{W}.$ The set of messages is denoted by $\mathcal{W}.$ We use $\Gamma_n$ as a transmission-code for this channel with the fixed block-length $n$ satisfying \eqref{fixedblocklength}.
Since $\mbf g \in \mc G_\theta,$ it follows that
\begin{align}
    \mbb P\left[ W\neq \hat{W}\right]= e(\Gamma_n,\mbf g)\leq  \theta. \nonumber
\end{align}
We have
\begin{align}
    H(W)&= \log \lvert \mathcal{W} \rvert \nonumber\\
    &=\log \lvert \Gamma_n \rvert \nonumber \\
    &\geq n\left(R_{\eta,\sup}+\frac{\epsilon}{2}\right), 
    \label{entropy1}
\end{align}
where we used \eqref{ratecondition} in the last step.
By applying Fano's inequality, we obtain
\begin{align}
    H(W|\hat{W})&\leq 1+ \mbb P\left[W\neq \hat{W}\right] \log \lvert \mathcal{W} \rvert \nonumber \\
    &\leq 1+\theta \log \lvert \mathcal{W} \rvert \nonumber \\
    &=1+\theta H(W). \nonumber
\end{align}

Now, on the one hand, it holds that
\begin{align}
    I(W;\hat{W})&=H(W)-H(W|\hat{W}) \nonumber \\
    &\geq (1-\theta)H(W)-1, \nonumber 
\end{align}
which yields
\begin{align}
   H(W)\leq \frac{1+I(W;\hat{W})}{1-\theta}. \nonumber
\end{align}
On the other hand, we have
\begin{align}
    \frac{1}{n} I(W;\hat{W})&\overset{(a)}{\leq} \frac{1}{n} I(\bs{T}^n;\bs{Z}^n) \nonumber \\
    &\overset{(b)}{=}\frac{1}{n}\sum_{i=1}^{n} I(\bs{Z}_i;\bs{T}^n|\bs{Z}^{i-1}) \nonumber \\
    &=\frac{1}{n} \sum_{i=1}^{n} h(\bs{Z}_i|\bs{Z}^{i-1})-h(\bs{Z}_i|\bs{T}^{n},\bs{Z}^{i-1}) \nonumber \\
    &\overset{(c)}{=} \frac{1}{n} \sum_{i=1}^{n} h(\bs{Z}_i|\bs{Z}^{i-1})-h(\bs{Z}_i|\bs{T}_i) \nonumber \\
    &\overset{(d)}{\leq} \frac{1}{n} \sum_{i=1}^{n} h(\bs{Z}_i)-h(\bs{Z}_i|\bs{T}_i) \nonumber \\
    &= \frac{1}{n} \sum_{i=1}^{n} I(\bs{T}_i,\bs{Z}_i)
    \nonumber \\
    &\leq  \sum_{i=1}^{n} \frac{1}{n} \log\det\left(\mbf I_{N_{R}}+\frac{1}{\sigma^2}\mbf g \mbf Q_{i} \mbf g^{H}\right) \nonumber \\
    &\overset{(e)}{\leq}\log\det\left(\frac{1}{n} \sum_{i=1}^{n}\left[\mbf I_{N_{R}}+\frac{1}{\sigma^2}\mbf g \mbf Q_{i} \mbf g^{H} \right]     \right) \nonumber \\
    &=\log\det\left(\mbf I_{N_{R}}+\frac{1}{\sigma^2}\mbf g \left(\frac{1}{n}\sum_{i=1}^{n}\mbf Q_{i}\right) \mbf g^{H}  \right) \nonumber \\
    &=\log\det\left(\mbf I_{N_{R}}+\frac{1}{\sigma^2}\mbf g \mbf Q^{\star} \mbf g^{H}  \right),
  \nonumber
\end{align} 
where $(a)$ follows from the Data Processing Inequality because $W\circlearrow{\bs{T}^n}\circlearrow{\bs{Z}^n}\circlearrow{\hat{W}}$ forms a Markov chain, $(b)$ follows from the chain rule of mutual information, (c) follows because
$\bs{T}_{1},\dots, \bs{T}_{i-1}, \bs{T}_{i+1},\dots, \bs{T}_{n},\bs{Z}^{i-1} \circlearrow{\bs{T}_{i}}\circlearrow{\bs{Z}_{i}}$ forms a Markov chain, $(d)$ follows because conditioning does not increase entropy and $(e)$ follows because $\log\circ\det$ is concave on the set of Hermitian positive semi-definite matrices, where $\mbf I_{N_{R}}+\frac{1}{\sigma^2}\mbf g \mbf Q_{i}\mbf g^H$ is Hermitian positive semi-definite for $i=1, \hdots,n.$

This yields
\begin{align}
H(W)\leq \frac{1+n\log\det(\mbf I_{N_{R}}+\frac{1}{\sigma^2}\mbf g \mbf Q^{\star} \mbf g^{H})}{1-\theta}. 
\label{entropy2}
\end{align}
The inequalities \eqref{entropy1} and \eqref{entropy2} imply that
\begin{align}
    n\left(R_{\eta,\sup}+\frac{\epsilon}{2}\right)&\leq \frac{1+n\log\det(\mbf I_{N_{R}}+\frac{1}{\sigma^2}\mbf g \mbf Q^{\star} \mbf g^{H})}{1-\theta} \nonumber \\
    &< \frac{1+n(R_{\eta,\sup}+\epsilon')}{1-\theta},
    \label{equivalent}
\end{align}
where we used that $\mbf g \in \mathcal{G}_\theta.$
The inequality \eqref{equivalent} is equivalent to
\begin{align}
    -\theta R_{\eta,\sup}+(1-\theta)\frac{\epsilon}{2}-\frac{1}{n}< \epsilon'. \nonumber
\end{align}
However, by the choice of $\theta$ and $n$, the left-hand side of this inequality is strictly larger than $\frac{\epsilon}{8},$ whereas $\epsilon'=\frac{\epsilon}{8}.$ This is a contradiction. 
 Thus \eqref{conversequation}  must be true. This proves the converse of Theorem \ref{cetathmMIMO}.

\section{Proof of Theorem \ref{ccretathmMIMO}}
\label{proofoutagecrcapacity}
\subsection{Direct Proof}
We extend the coding scheme provided in \cite{part2} to MIMO slow fading channels. By continuity, it suffices to show that 
$$ \underset{ \substack{U \\{\substack{U \circlearrow{X} \circlearrow{Y}\\ I(U;X)-I(U;Y) \leq C'}}}}{\max} I(U;X)  $$ is an achievable $\eta$-outage CR rate for every $C'<C_{\eta}(P,N_{T}\times N_{R}).$
Let $U$ be a random variable satisfying $U \circlearrow{X} \circlearrow{Y}$ and $I(U;X)-I(U;Y) \leq C'$. Let $\alpha,\delta>0$ and $0\leq \eta<1.$ We are going to show that $H=I(U;X)$ is an achievable $\eta$-outage CR rate. Without loss of generality, assume that the distribution of $U$ is a possible type for block-length $n$.
For some $\mu>0,$ we let
{{\begin{align}
N_{1}&=\lfloor 2^{n[I(U;X)-I(U;Y)+3\mu]} \rfloor \nonumber\\
N_{2}&=\lfloor 2^{n[I(U;Y)-2\mu]}\rfloor. \nonumber
\end{align}}}For each pair $(i,j)$ with $1\leq i \leq N_{1}$ and $1\leq j \leq N_{2}$, we define a random sequence $\bs{U}_{i,j}\in\mathcal{U}^n$ of type $P_{U}$. Each realization $\bs{u}_{i,j}$ of $\bs{U}_{i,j}$ is known to both terminals.  
This means that  $N_{1}$ codebooks $C_{i}, 1\leq i \leq N_{1}$, are known to both terminals, where each codebook contains $N_{2}$ sequences, $ \bs{u}_{i,j}, \ j=1,\hdots, N_2$.

It holds for every $X$-typical $\bs{x}$ that $$\mbb P[\exists (i,j) \ \text{s.t} \ \bs{U}_{ij} \in \mathcal{T}_{U|X}^{n}\left(\bs{x}\right)|X^{n}=\bs{x}]\geq 1-2^{-2^{nc'}},$$ for a suitable $c'>0,$ as in the proof of  Theorem 4.1 in \cite{part2}. 
 For $\Phi(\bs{x})$, we choose a sequence $\bs{u}_{ij}$ jointly typical with $\bs{x}$ (either one if there are several).   Let $f_1(\bs{x})=i$ if $\Phi(\bs{x}) \in C_{i}$.  If no such $\bs{u}_{i,j}$ exists, then $f_1(\bs{x})=N_1+1$ and $\Phi(\bs{x})$ is set to a constant sequence $\bs{u}_0$ different from all the $\bs{u}_{ij}s$ and known to both terminals.
  Since $ C'<C_{\eta}(P,N_T\times N_R)$, we choose $\mu$ to be sufficiently small such that
      \begin{align}
     \frac{\log \lVert f_1 \rVert}{n}&=\frac{\log(N_1+1)}{n} \nonumber \\
     &\leq C_{\eta}(P,N_{T}\times N_{R})-\mu',
     \label{inequalitylogfSISO}
      \end{align}
for some $\mu'>0,$
 where $\lVert f_1 \rVert$ refers to the cardinality of the set of messages $\{i^\star=f_1(\bs{x})\}
 \footnote{\text{This is the same notation used in} \cite{codingtheorems}.}$.
 The message $i^\star=f_1(\bs{x})$, with $i^\star\in\{1,\hdots,N_1+1\}$, is encoded to a sequence $\mbf t$ using a code sequence $(\Gamma^\star_n)_{n=1}^{\infty}$ with rate $\frac{\log \lvert \Gamma^\star_n \rvert}{n}=\frac{\log \lVert f_1 \rVert}{n}$ satisfying \eqref{inequalitylogfSISO}
 and with error probability $e(\Gamma^\star_n,\mbf G)$ satisfying for sufficiently large $n$
 \begin{align}
     \mbb P\left[e(\Gamma^\star_n,\mbf G) \leq \theta \right] \geq 1-\eta,
     \label{erroroutage}
 \end{align}
where $\theta$ is  a positive constant satisfying $\theta\leq \frac{\alpha}{2}.$ 
  \color{black}
  From the definition of the $\eta$-outage capacity, we know that such a code sequence exists. The sequence $\mbf t$ is sent over the MIMO slow fading channel. Let $\mbf z$ be the channel output sequence. Terminal $B$ decodes the message $\tilde{i}^\star$ from the knowledge of $\mbf z.$
Let $\Psi(\bs{y},\mbf z)=\bs{u}_{\tilde{i}^\star,j}$ if $\bs{u}_{\tilde{i}^\star,j}$ and $\bs{y}$ are jointly typical . If there is no such $\bs{u}_{\tilde{i}^\star,j}$ or there are several, we set $\Psi(\bs{y},\mbf z)=\bs{u}_0$ (since $K$ and $L$ must have the same alphabet).
Now, we are going to show that the requirements in $\eqref{errorMIMOcorrelated},$  $\eqref{cardinalityMIMOcorrelated}$ and $\eqref{rateMIMOcorrelated}$ are satisfied.
Clearly, (\ref{cardinalityMIMOcorrelated}) is satisfied  for $c=H(X)+\mu+1$  because
{{\begin{align}
|\mathcal{K}|&=N_1 N_2+1 \nonumber \\
             &\leq  2^{n\left[I(U;X)+\mu\right]}+1 \nonumber \\
             &\leq2^{n\left[H(X)+\mu+1\right]}.\nonumber
\end{align}}}We define next for any $(i,j)\in \{1,\hdots,n\}\times\{1,\hdots,n\}$ and for any fixed realization  $\bs{u}_{i,j}$ of $\bs{U}_{i,j}$ the set
$$\mc S=\{ \bs{x}\in\mathcal{X}^{n} \ \text{s.t.} \ (\bs{x},\bs{u}_{i,j}) \ \text{jointly} \ \text{typical}\}.$$
Then, it holds that 
\begin{align}
\mbb P[K=\bs{u}_{i,j}] &=\sum_{\bs{x}\in\mc S}\mbb P[K=\bs{u}_{i,j}|X^n=\bs{x}]P_{X}^n(\bs{x}) +\sum_{\bs{x}\in\mc S^c}\mbb P[K=\bs{u}_{i,j}|X^n=\bs{x}]P_{X}^n(\bs{x}) \nonumber \\
&\overset{(\RM{1})}{=}\sum_{\bs{x}\in\mc S}\mbb P[K=\bs{u}_{i,j}|X^n=\bs{x}]P_{X}^n(\bs{x}) \nonumber \\
&\leq \sum_{\bs{x}\in\mc S}P_{X}^n(\bs{x}) \nonumber \\
&=P_{X}^{n}(\{\bs{x}: (\bs{x},\bs{u}_{i,j}) \ \text{jointly} \ \text{typical}\}) \nonumber\\
& = 2^{-nI(U;X)+\kappa(n)}, \nonumber
\end{align}
for some $\kappa(n)$ with $\underset{n\rightarrow \infty}{\lim} \frac{\kappa(n)}{n}=0$,
where (\RM{1}) follows because for  $(\mathbf{x},\bs{u}_{i,j})$ being not jointly typical, we have $\mbb P[K=\bs{u}_{i,j}|X^n=\bs{x}]=0.$ This yields
{{\begin{align}
H(K)\geq nI(U;X)-\kappa'(n)
\nonumber \end{align}}}
for some $\kappa'(n)>0$ with $\underset{n\rightarrow \infty}{\lim} \frac{\kappa'(n)}{n}=0.$
Therefore, for sufficiently large $n,$ it holds that
\begin{align}
    \frac{H(K)}{n}>H-\delta. \nonumber
\end{align}
Thus, (\ref{rateMIMOcorrelated}) is satisfied.

 Now, it remains to prove that \eqref{errorMIMOcorrelated} is satisfied. For this purpose, we define $\mbf M=\bs{U}_{11},\hdots, \bs{U}_{N_{1}N_{2}}$ to be the joint random variable of all $\bs{U}_{i,j}s.$ We further define the following two sets which depend on $\mbf M$:
\begin{align}
    S_{1}(\mbf M)&=\{(\bs{x},\bs{y}):(\Phi(\bs{x}),\bs{x},\bs{y}) \in \mathcal{T}_{UXY}^{n}\} \nonumber
\end{align} and
\begin{align}
    S_{2}(\mbf M)=\{(\bs{x},\bs{y}):(\bs{x},\bs{y}) \in S_{1}(\mbf M) \ \text{s.t.} \ \exists \ \bs{U}_{i\ell}\neq\bs{U}_{ij}=\Phi(\bs{x}) \nonumber \\   \text{jointly} \ \text{typical with} \ \bs{y} \ (\text{with the same first index} \ i)
\}.\nonumber
\end{align}
It is proved in \cite{part2} that
\begin{align}
    \mathbb{E}_{\mbf M}\left[P_{XY}^{n}(S_{1}^{c}(\mbf M))+P_{XY}^{n}(S_{2}(\mbf M))\right] \leq \beta,
    \label{averagebeta}
\end{align}
where $\beta \leq \frac{\alpha}{2}$ for sufficiently large $n$. \begin{remark}
$P_{XY}^{n}(S_{1}^{c}(\mbf M))$ and $P_{XY}^{n}(S_{2}(\mbf M))$ are here random variables depending on $\mbf M.$
\end{remark}We choose a realization $\mbf m=\bs{u}_{11},\hdots, \bs{u}_{N_1N_2}$ satisfying:
\begin{align}
    P_{XY}^{n}(S_{1}^{c}(\mbf m))+P_{XY}^{n}(S_{2}(\mbf m)) \leq \beta. \nonumber
\end{align} 
From \eqref{averagebeta}, we know that such a realization exists.
Now, we define the following event:
\begin{align}
    \mathcal{D}_{\mbf m}= ``\Phi(X^n) \ \text{is equal to none of the} \  \bs{u}_{i,j}s". \nonumber
\end{align}
We denote its complement by $\mc D_{\mbf m}^{c}.$
We further define $I^\star=f_1(X^n)$ to be the random message generated by Terminal $A$ and  $\tilde{I}^\star$ to be the random message decoded by Terminal $B$. 
We have
\begin{align}
    \mbb P[K\neq L|\mbf G] \nonumber &=\mbb P[K\neq L|\mbf G,I^\star=\tilde{I}^\star]\mbb P[I^\star=\tilde{I}^\star|\mbf G] + \mbb P[K\neq L|\mbf G,I^\star\neq \tilde{I}^\star]\mbb P[I^\star\neq\tilde{I}^\star|\mbf G] \nonumber \\
        &\leq \mbb P[K\neq L|\mbf G,I^\star=\tilde{I}^\star]+ \mbb P[I^\star\neq\tilde{I}^\star|\mbf G].\nonumber
\end{align}
Here,
\begin{align}
    \mbb P[K\neq L|\mbf G,I^\star=\tilde{I}^\star] 
   &= \mbb P[K\neq L|\mbf G,I^\star=\tilde{I}^\star,\mathcal{D}_{\mbf m}]\mbb P[\mathcal{D}_{\mbf m}|\mbf G,I^\star=\tilde{I}^\star] + \mbb P[K\neq L|\mbf G,I^\star=\tilde{I}^\star,\mathcal{D}_{\mbf m}^c]\mbb P[\mathcal{D}_{\mbf m}^c|\mbf G,I^\star=\tilde{I}^\star] \nonumber \\
   &\overset{(\RM{1})}{=}\mbb P[K\neq L|\mbf G,I^\star=\tilde{I}^\star,\mathcal{D}_{\mbf m}^c]\mbb P[\mathcal{D}_{\mbf m}^c|\mbf G,I^\star=\tilde{I}^\star] \nonumber \\
   &\leq \mbb P[K\neq L|\mbf G,I^\star=\tilde{I}^\star,\mathcal{D}_{\mbf m}^c],\nonumber
\end{align}
where $(\RM{1})$ follows from $\mbb P[K\neq L|\mbf G,I^\star=\tilde{I}^\star,\mathcal{D}_{\mbf m}]=0,$ since conditioned on $\mbf G$, $I^\star=\tilde{I}^\star$ and $\mathcal{D}_{\mbf m}$, we know that $K$ and $L$ are both equal to $\bs{u}_0$.
It follows that
\begin{align}
    \mbb P[K\neq L|\mbf G] 
    &\leq \mbb P[K\neq L|\mbf G,I^\star=\tilde{I}^\star,\mathcal{D}_{\mbf m}^c]+ \mbb P[I^\star\neq\tilde{I}^\star|\mbf G] \nonumber \\
    &\leq P_{XY}^n\left(S_{1}^{c}(\mbf m)\cup S_{2}(\mbf m)\right)+ \mbb P[I^\star\neq\tilde{I}^\star|\mbf G] \nonumber \\
    &\overset{(a)}{=}P_{XY}^{n}(S_{1}^{c}(\mbf m))+P_{XY}^n\left(S_{2}(\mbf m)\right) +\mbb P[I^\star\neq\tilde{I}^\star|\mbf G] \nonumber \\
    &\leq \beta+ \mbb P[I^\star\neq\tilde{I}^\star|\mbf G],\nonumber
\end{align}
where $(a)$ follows because $S_{1}^{c}(\mbf m)$ and $S_{2}(\mbf m)$ are disjoint. 
It holds that
\begin{align}
    &\mbb P\left[I^\star\neq \tilde{I}^\star|\mbf G\right]\leq \theta \implies \mbb P[K\neq L|\mbf G] \leq \beta+ \theta. \nonumber
\end{align}
Since $\beta+ \theta\leq \alpha$, it follows that
\begin{align}
    &\mbb P\left[I^\star\neq \tilde{I}^\star|\mbf G\right]\leq \theta  \implies \mbb P[K\neq L|\mbf G] \leq \alpha. \nonumber
\end{align}
From \eqref{erroroutage}, we know that
\begin{align}
   \mbb P\left[ \mbb P\left[I^\star\neq \tilde{I}^\star|\mbf G\right]\leq \theta\right] \geq 1-\eta. \nonumber
\end{align}
Thus
\begin{align}
    \mbb P\left[  \mbb P[K\neq L|\mbf G] \leq \alpha          \right] &\geq \mbb P\left[\mbb P\left[ I^\star\neq \tilde{I}^\star|\mbf G     \right]\leq \theta\right] \nonumber \\
    &\geq 1-\eta. \nonumber
\end{align} 
This completes the direct proof of Theorem \ref{ccretathmMIMO}.
\subsection{Converse Proof}
Let $(K,L)$ be a permissible pair according to the fixed CR-generation protocol of block-length $n$ introduced in Section \ref{systemmodel}. We recall that the latter consists of:
\begin{enumerate}
    \item A function $\Phi$ that maps $X^n$ into a random variable $K$ with alphabet $\mathcal{K}$ generated by Terminal $A.$
    \item A function $\Lambda$ that maps $X^n$ into the input sequence $\bs{T}^n \in \mbb C^{N_T\times n}$  satisfying the following power constraint
    \begin{equation}
\frac{1}{n}\sum_{i=1}^{n}\bs{T}_{i}^H\bs{T}_{i}\leq P, \quad \text{almost surely}.   \nonumber
\end{equation}
    \item A function $\Psi$ that maps $Y^n$ and the  output sequence $\bs{Z}^n \in \mbb C^{N_R\times n}$ into a random variable $L$ with alphabet $\mathcal{K}$ generated by Terminal $B.$
\end{enumerate}
We further assume that $(K,L)$ satisfies $\eqref{errorMIMOcorrelated}$  $\eqref{cardinalityMIMOcorrelated}$ and $\eqref{rateMIMOcorrelated}$, where the maximum error probability $\alpha>0$ in \eqref{errorMIMOcorrelated} and the constant $\delta>0$ in \eqref{rateMIMOcorrelated} are fixed arbitrarily.
We are going to show for some $\alpha'(n)>0$ that
\begin{align}
    \frac{H(K)}{n} \leq \underset{ \substack{U \\{\substack{U \circlearrow{X} \circlearrow{Y}\\ I(U;X)-I(U;Y) \leq C_{\eta}(P,N_{T} \times N_{R})+\alpha'(n)}}}}{\max} I(U;X), \nonumber 
\end{align}
where 
\begin{align}
&C_{\eta}(P,N_{T} \times N_{R}) =\sup \ \Big\{R: \underset{\mbf Q \in \mathcal{Q}_{P}}{\inf }\mbb P\left[\log\det(\mathbf{I}_{N_{R}}+\frac{1}{\sigma^2}\mathbf{G}\mathbf{Q}\mathbf{G}^{H})< R \right] \leq \eta\Big\} \nonumber
\end{align}
and where $\underset{n\rightarrow \infty}{\lim}\alpha'(n)$ can be made arbitrarily small for a suitable choice of $\alpha>0$ and some other constant $\epsilon>0.$
\color{black}
In our proof, we will use  the following lemma: 
\begin{lemma} (Lemma 17.12 in \cite{codingtheorems})
For arbitrary random variables $S$ and $R$ and sequences of random variables $X^{n}$ and $Y^{n}$, it holds that
\begin{align}
 I(S;X^{n}|R)-I(S;Y^{n}|R)  
 &=\sum_{i=1}^{n} I(S;X_{i}|X_{1},\dots, X_{i-1}, Y_{i+1},\dots, Y_{n},R) \nonumber \\ &\quad -\sum_{i=1}^{n} I(S;Y_{i}|X_{1},\dots, X_{i-1}, Y_{i+1},\dots, Y_{n},R) \nonumber \\
 &=n[I(S;X_{J}|V)-I(S;Y_{J}|V)],\nonumber
\end{align}
where $V=X_{1},\dots, X_{J-1},Y_{J+1},\dots, Y_{n},R,J$, with $J$ being a random variable independent of $R$,\ $S$, \ $X^{n}$ \ and $Y^{n}$ and uniformly distributed on $\{1 ,\dots, n \}$.
\label{lemma1}
\end{lemma}Let $J$ be a random variable uniformly distributed on $\{1,\dots, n\}$ and independent of $K$, $X^n$ and $Y^n$. We further define $U=K,X_{1},\dots, X_{J-1},Y_{J+1},\dots, Y_{n},J.$ It holds that $U \circlearrow{X_J} \circlearrow{Y_J}.$ \\
Notice  that
{{\begin{align}
H(K)&=I(K;X^{n}) \nonumber\\
&\overset{(\RM{1})}{=}\sum_{i=1}^{n} I(K;X_{i}|X_{1},\dots, X_{i-1}) \nonumber\\
&=n I(K;X_{J}|X_{1},\dots, X_{J-1},J) \nonumber\\
&\overset{(\RM{2})}{\leq }n I(U;X_{J}), \nonumber
\end{align}}}where $(\RM{1})$ and $(\RM{2})$ follow from the chain rule for mutual information.\\
We will show next that for some $\alpha'(n)>0$
\begin{align}
    I(U;X_J)-I(U;Y_J) \leq C_{\eta}(P,N_{T} \times N_{R})+\alpha'(n), \nonumber
\end{align}
where $\underset{n\rightarrow\infty}{\lim} \alpha'(n)$ can be made arbitrarily small.
Applying Lemma \ref{lemma1} for $S=K$, $R=\varnothing$ with $V=X_1,\hdots, X_{J-1},Y_{J+1},\hdots, Y_{n},J$ yields
{{\begin{align}
I(K;X^{n})-I(K;Y^{n}) 
&=n[I(K;X_{J}|V)-I(K;Y_{J}|V)] \nonumber\\
&\overset{(a)}{=}n[I(KV;X_{J})-I(K;V)-I(KV;Y_{J})+I(K;V)] \nonumber\\
&\overset{(b)}{=}n[I(U;X_{J})-I(U;Y_{J})], 
\label{UhilfsvariableMIMO1}
\end{align}}}where $(a)$ follows from the chain rule for mutual information and $(b)$ follows from $U=K,V$. \\
It results using (\ref{UhilfsvariableMIMO1}) that
{{\begin{align}
n[I(U;X_{J})-I(U;Y_{J})]
&=I(K;X^{n})-I(K;Y^{n}) \nonumber\\
&=H(K)-H(K|X^{n})-I(K;Y^{n}) \nonumber\\
&\overset{(c)}{=}H(K)-I(K;Y^{n})\nonumber \\ 
&=H(K|Y^n)
\label{star2MIMO2}
\end{align}}}where $(c)$ follows because  $K=\Phi(X^{n}).$
\\
Next, we will show for some $\alpha'(n)>0$ that
\begin{align}
    \frac{H(K|Y^n)}{n}\leq C_{\eta}(P,N_{T}\times N_{R})+\alpha'(n). \nonumber
\end{align}

Let $\text{cov}(\bs{T}_i)=\mathbf{Q}_i$   for $i=1, \hdots,n,$ where $\bs{T}_{i}\in \mbb C^{N_T}, i=1, \hdots,n.$ We define
\begin{align}
   \mbf Q^\star=\frac{1}{n}\sum_{i=1}^{n} \mbf Q_{i}. \nonumber
\end{align}
By Lemma \ref{traceQ}, we know that $\text{tr}(\mbf Q^\star)\leq P$ and therefore $\mbf Q^\star \in \mc Q_P.$
Let
\begin{align}
&R(\mbf Q^\star)= \sup \Big\{R: \mbb P\left[\log\det(\mathbf{I}_{N_{R}}+\frac{1}{\sigma^2}\mathbf{G}\mbf Q^\star\mathbf{G}^{H})< R\right]\leq \eta \Big\}. \nonumber
\end{align}
Since $\mbf Q^\star \in \mc Q_P$, Lemma \ref{supQ} implies that
\begin{align}
 R(\mbf Q^\star)\leq C_{\eta}(P,N_{T} \times N_{R}).
 \label{comparerates}
 \end{align}
We consider for $\epsilon>0$ the set
\begin{align}
    &\Omega=\Big\{\mathbf{g} \in \mathbb{C}^{N_{R}\times N_{T}}:   \mbb P\left[K\neq L|\mathbf{G}=\mathbf{g}\right]\leq \alpha \ \text{and} \ \log\det(\mathbf{I}_{N_{R}}+\frac{1}{\sigma^2}\mathbf{g}\mbf Q^\star\mathbf{g}^{H}) \leq   R(\mbf Q^\star)+\epsilon \Big\}. \nonumber
\end{align}
\begin{lemma}
\begin{align}
\mbb P\left[\mbf G \in \Omega\right]>0. \nonumber 
\end{align}
\label{probOmega}
\end{lemma}
\begin{proof}
We define $\forall R \geq 0$ and $\forall \mbf Q \in \mc Q_P$ the following set:
\begin{align}
    &\mathcal{F}(\mbf Q, R)=\Big\{\mbf g \in \mbb C^{N_R\times N_T}:\log\det(\mathbf{I}_{N_{R}}+\frac{1}{\sigma^2}\mathbf{g}\mbf Q \mathbf{g}^{H}) \leq R \Big\}. \nonumber
\end{align}
\begin{remark}
Unlike in the set $\mc E(\mbf Q, R)$ defined in \eqref{setE}, we do not impose strict inequality in $\mathcal{F}(\mbf Q,R).$
\end{remark}
From the definition of $R(\mbf Q^\star)$, we have
\begin{align}
    \eta &< \mbb P \left[\mbf G \in \mc E(\mbf Q^\star,R(\mbf Q^\star)+\epsilon)\right] \nonumber \\
    &\leq \mbb P\left[\mbf G \in \mc F(\mbf Q^\star,R(\mbf Q^\star)+\epsilon) \right]. \nonumber
\end{align}

Then, it holds that
\begin{align}
    \mbb P\left[\mbf G \in \mc F(\mbf Q^\star,R(\mbf Q^\star)+\epsilon)\right] =\eta_1, \nonumber
\end{align}
where $0\leq \eta<\eta_1\leq 1.$

 It follows using \eqref{errorMIMOcorrelated} that
\begin{align}
    1-\eta 
    & \leq \mbb P\left[\mbb P\left[K\neq L|\mbf G\right]\leq \alpha\right] \nonumber \\
    &=\mbb P\left[ \mbb P\left[K\neq L\bigm|\mbf G\right]\leq \alpha\bigm|\mbf G \in \mc F(\mbf Q^\star,R(\mbf Q^\star)+\epsilon) \right] \mbb P\left[\mbf G \in \mc F(\mbf Q^\star,R(\mbf Q^\star)+\epsilon)\right] \nonumber \\
    &\quad+\mbb P\left[ \mbb P\left[K\neq L\bigm|\mbf G\right]\leq \alpha\bigm|{\mbf G \notin \mc F(\mbf Q^\star,R(\mbf Q^\star)+\epsilon)} \right] \mbb P\left[{\mbf G \notin \mc F(\mbf Q^\star,R(\mbf Q^\star)+\epsilon)}\right]\nonumber \\
    &=\eta_1 \ \mbb P\left[ \mbb P\left[K\neq L\bigm|\mbf G\right]\leq \alpha\bigm|\mbf G \in \mc F(\mbf Q^\star,R(\mbf Q^\star)+\epsilon) \right]\nonumber \\
    &\quad+(1-\eta_1) \ \mbb P\left[ \mbb P\left[K\neq L\bigm|\mbf G\right]\leq \alpha\bigm|{\mbf G \notin \mc F(\mbf Q^\star,R(\mbf Q^\star)+\epsilon)} \right]
    \nonumber \\
    &\leq \eta_1 \ \mbb P\left[ \mbb P\left[K\neq L\bigm|\mbf G\right]\leq \alpha\bigm|\mbf G \in \mc F(\mbf Q^\star,R(\mbf Q^\star)+\epsilon) \right]+(1-\eta_1) \nonumber \\
    &\leq \mbb P\left[ \mbb P\left[K\neq L\bigm|\mbf G\right]\leq \alpha\bigm|\mbf G \in \mc F(\mbf Q^\star,R(\mbf Q^\star)+\epsilon) \right]+(1-\eta_1)\nonumber \\
    &< \mbb P\left[ \mbb P\left[K\neq L\bigm|\mbf G\right]\leq \alpha\bigm|\mbf G \in \mc F(\mbf Q^\star,R(\mbf Q^\star)+\epsilon) \right]+(1-\eta), \nonumber
\end{align}
where we used that $1-\eta_1<1-\eta.$ This means that
\begin{align}
    \mbb P\left[ \mbb P\left[K\neq L\bigm|\mbf G\right]\leq \alpha\bigm|\mbf G \in \mc F(\mbf Q^\star,R(\mbf Q^\star)+\epsilon) \right]>0. \nonumber
\end{align}
In addition, since $\eta_1>0$, we have
\begin{align}
    \mbb P\left[ \mbb P\left[K\neq L\bigm|\mbf G\right]\leq \alpha,\mbf G \in \mc F(\mbf Q^\star,R(\mbf Q^\star)+\epsilon) \right]>0. \nonumber
\end{align}
Thus
$$\mbb P\left[\mbf G \in \Omega\right]>0.$$
\end{proof}Next, we define $\tilde{\mbf G}$ to be a  random matrix, independent of $X^n$,$Y^n$ and $\bs{\xi}^n$, with alphabet $\Omega$ such that for every Borel set $\seta \subseteq \mbb C^{N_R\times N_T},$ it holds that
\begin{align}
    \mbb P \left[ \tilde{\mbf G}\in \seta \right]=\mbb P \left[ \mbf G\in\seta|\mbf G\in\Omega\right]. \nonumber
\end{align}
By Lemma \ref{probOmega}, we know that such a $\tilde{\mbf G}$ exists.

We fix the CR generation protocol and change the state distribution of the slow fading channel. We obtain the following new MIMO channel:
\begin{align}
    \tilde{\bs Z}_{i}=\tilde{\mbf G}\bs{T}_i+\bs{\xi}_i \quad i=1, \hdots,n,  \nonumber
\end{align}
where $\tilde{\bs Z}^n$ is the new output sequence.
We further define $\tilde{L}$ such that
\begin{align}
    \tilde{L}=\Psi(Y^n,\tilde{\bs Z}^n). \nonumber
\end{align}
Clearly, it holds for any $\mbf g \in \Omega$ that
\begin{align}
     \mbb P\left[K\neq \tilde{L}|\tilde{\mbf G}= \mbf g\right]\leq \alpha   \label{newerrorinequality}
\end{align}
and that
\begin{align}
    \log\det(\mathbf{I}_{N_{R}}+\frac{1}{\sigma^2}\mathbf{g}\mbf Q^\star\mathbf{g}^{H}) \leq R(\mbf Q^\star)+\epsilon.
    \label{absolutevalue}
\end{align}
Furthermore, since $\bs{\xi}_i \sim \mathcal{N}_{\mathbb{C}}(\bs{0}_{N_R},\sigma^2 \mathbf{I}_{N_R}), i=1, \hdots,n$, it follows for $i=1, \hdots,n$ that
\begin{align}
    I(\bs{T}_{i};\tilde{\bs Z}_{i}|\tilde{\mbf G}= \mbf g)&\leq \log\det(\mathbf{I}_{N_{R}}+\frac{1}{\sigma^2}\mathbf{g}\mathbf{Q}_{i} \mathbf{g}^{H})\quad \forall \mbf g\in\Omega.
    \label{mutualinfmax}
\end{align}
We recall that the goal is to prove that for some $\alpha'(n)>0$ \begin{align}
    \frac{H(K|Y^n)}{n}\leq C_{\eta}(P,N_{T}\times N_{R})+\alpha'(n). \nonumber
\end{align}

Now, we have
\begin{align}
    \frac{1}{n}H(K|Y^n)&=\frac{1}{n}H(K|\tilde{\mbf G},Y^n) \nonumber \\
            &=\frac{1}{n}H(K|\tilde{\mbf G},Y^n,\tilde{\bs Z}^n)+\frac{1}{n}I(K;\tilde{\bs Z}^n|\tilde{\mbf G},Y^n), \nonumber
\end{align}
where we used that $\tilde{\mbf G}$ is independent of $(K,Y^n).$ 
On the one hand, we have
\begin{align}
                                \frac{1}{n}H\left(K|\tilde{\bs Z}^n,\tilde{\mbf G},Y^n\right)  &\overset{(a)}{\leq } \frac{1}{n} H\left(K|\tilde{L},\tilde{\mbf G}\right) \nonumber \\    &\overset{(b)}{\leq } \frac{1}{n}\mbb E\left[1+\log|\mathcal{K}|\mbb P[K\neq \tilde{L}|\tilde{\mbf G}]\right]  \nonumber \\                          &=\frac{1}{n}+\frac{1}{n} \log|\mathcal{K}|\mbb E \left[ P[K\neq \tilde{L}|\tilde{\mbf G}]\right]  \nonumber \\
                                &\overset{(c)}{\leq } \frac{1}{n}+\frac{1}{n}\alpha \log|\mathcal{K}|  \nonumber \\                               &\overset{(d)}{\leq } \frac{1}{n}+\alpha \ c,  \nonumber
\end{align}
where (a) follows from $\tilde{L}=\Psi(Y^n,\tilde{\bs Z}^n)$, (b) follows from Fano's Inequality, (c) follows from \eqref{newerrorinequality}   and (d) follows from $\log|\mathcal{K}|\leq cn$ in \eqref{cardinalityMIMOcorrelated}.
On the other hand, we have
 \begin{align} 
\frac{1}{n}I(K;\tilde{\bs Z}^n|\tilde{\mbf G},Y^{n}) &\leq \frac{1}{n} I(X^{n},K;\tilde{\bs Z}^n|\tilde{\mbf G},Y^{n}) \nonumber\\
& \overset{(a)}{\leq }\frac{1}{n} I(\bs{T}^n;\tilde{\bs Z}^n|\tilde{\mbf G},Y^{n})  \nonumber \\
& = \frac{1}{n} \left[h(\tilde{\bs Z}^n|\tilde{\mbf G},Y^{n})- h(\tilde{\bs Z}^n|\bs{T}^n,\tilde{\mbf G},Y^{n}) \right]\nonumber \\
& \overset{(b)}{=} \frac{1}{n} \left[h(\tilde{\bs Z}^n|\tilde{\mbf G},Y^{n})- h(\tilde{\bs Z}^n|\tilde{\mbf G},\bs{T}^n)\right] \nonumber \\
& \overset{(c)}{\leq }  \frac{1}{n} \left[h(\tilde{\bs Z}^n|\tilde{\mbf G})- h(\tilde{\bs Z}^n|\tilde{\mbf G},\bs{T}^n)\right] \nonumber \\
& =\frac{1}{n} I(\bs{T}^n;\tilde{\bs Z}^n|\tilde{\mbf G})  \nonumber \\
& \overset{(d)}{=} \frac{1}{n}\sum_{i=1}^{n} I(\tilde{\bs Z}_{i};\bs{T}^n|\tilde{\mbf G},\tilde{\bs Z}^{i-1}) \nonumber \\
& = \frac{1}{n}\sum_{i=1}^{n} h(\tilde{\bs Z}_{i}|\tilde{\mbf G},\tilde{\bs Z}^{i-1})-h(\tilde{\bs Z}_{i}|\tilde{\mbf G},\bs{T}^n,\tilde{\bs Z}^{i-1}) \nonumber \\
& \overset{(e)}{=} \frac{1}{n}\sum_{i=1}^{n} h(\tilde{\bs Z}_{i}|\tilde{\mbf G},\tilde{\bs Z}^{i-1})-h(\tilde{\bs Z}_{i}|\tilde{\mbf G},\bs{T}_{i}) \nonumber \\
& \overset{(f)}{\leq} \frac{1}{n} \sum_{i=1}^{n} h(\tilde{\bs Z}_{i}|\tilde{\mbf G})-h(\tilde{\bs Z}_{i}|\tilde{\mbf G},\bs{T}_{i}) \nonumber \\
& = \frac{1}{n}\sum_{i=1}^{n} I(\bs{T}_{i};\tilde{\bs Z}_{i}|\tilde{\mbf G}) \nonumber \\
&\overset{(g)}{\leq}\frac{1}{n} \sum_{i=1}^{n}\mbb E \left[\log\det(\mathbf{I}_{N_{R}}+\frac{1}{\sigma^2}\tilde{\mathbf{G}}\mathbf{Q}_{i} \tilde{\mathbf{G}}^{H})\right]\nonumber \\
&=\mbb E\left[\frac{1}{n}\sum_{i=1}^{n}\log\det(\mathbf{I}_{N_{R}}+\frac{1}{\sigma^2}\tilde{\mathbf{G}}\mathbf{Q}_{i} \tilde{\mathbf{G}}^{H})    \right] \nonumber \\
&\overset{(h)}{\leq} \mbb E\left[\log\det\left(\frac{1}{n}\sum_{i=1}^{n} \left[\mathbf{I}_{N_{R}}+\frac{1}{\sigma^2}\tilde{\mathbf{G}}\mathbf{Q}_{i} \tilde{\mathbf{G}}^{H}\right]\right)\right] \nonumber \\
&=\mbb E\left[\log\det\left(\mathbf{I}_{N_{R}}+\frac{1}{\sigma^2}\tilde{\mathbf{G}}\left(\frac{1}{n}\sum_{i=1}^{n}\mathbf{Q}_{i}\right)\tilde{\mathbf{G}}^{H}\right)\right] \nonumber \\
&=\mbb E\left[\log\det\left(\mathbf{I}_{N_{R}}+\frac{1}{\sigma^2}\tilde{\mathbf{G}}\mbf Q^\star\tilde{\mathbf{G}}^{H}\right)\right] \nonumber \\
& \overset{(i)}{\leq}  R(\mbf Q^\star)+\epsilon \nonumber \\
&\overset{(j)}{\leq} C_{\eta}(P,N_{T} \times N_{R})+\epsilon,\nonumber
  \end{align}where $(a)$ follows from the Data Processing Inequality because $Y^{n}\circlearrow{X^{n}K}\circlearrow{\tilde{\mbf G}\bs{T}^n}\circlearrow{\tilde{\bs Z}^{n}}$ forms a Markov chain, $(b)$ follows because $Y^{n}\circlearrow{X^{n}K}\circlearrow{\tilde{\mbf G}\bs{T}^n}\circlearrow{\tilde{\bs Z}^{n}}$ forms a Markov chain, $(c)(f)$ follow because conditioning does not increase entropy, $(d)$ follows from the chain rule for mutual information, $(e)$ follows because $\bs{T}_{1},\dots, \bs{T}_{i-1},\bs{T}_{i+1},\dots, \bs{T}_{n},\tilde{\bs Z}^{i-1} \circlearrow{\tilde{\mbf G},\bs{T}_{i}}\circlearrow{\tilde{\bs Z}_{i}}$ forms a Markov chain, $(g)$ follows from \eqref{mutualinfmax},
 $(h)$ follows from Jensen's Inequality since the function $\log\circ\det$ is concave on the set of Hermitian positive semi-definite matrices and since $\mbf I_{N_{R}}+\frac{1}{\sigma^2}\tilde{\mbf G}\mbf Q_{i}\tilde{\mbf G}^H$ is Hermitian positive semi-definite for $i=1, \hdots,n$ and $(i)$ follows from \eqref{absolutevalue} and $(j)$ follows from \eqref{comparerates}. \\
This proves that for $0\leq \eta<1$
 \begin{align}
     \frac{H(K|Y^n)} {n} \leq C_{\eta}(P,N_{T} \times N_{R})+\alpha'(n),
     \label{star1MIMOnonemptyMIMO}
 \end{align}
 where $\alpha'(n)=\frac{1}{n}+\alpha c+\epsilon >0.$ 
 \\~\\
From (\ref{star2MIMO2}) and (\ref{star1MIMOnonemptyMIMO}), we deduce that for $0\leq\eta<1$
{{\begin{align}
&I(U;X_{J})-I(U;Y_{J})  \leq C_{\eta}(P,N_{T} \times N_{R}) +\alpha'(n), \nonumber
\end{align}}}where $U \circlearrow{X_{J}} \circlearrow{Y_{J}}.$ \\
\color{black}Since the joint distribution of $X_{J}$ and $Y_{J}$ is equal to $P_{XY}$, $\frac{H(K)}{n}$ is upper-bounded by $I(U;X)$ subject to $I(U;X)-I(U;Y) \leq C_{\eta}(P,N_{T} \times N_{R}) + \alpha'(n)$ with $U$ satisfying $U \circlearrow{X} \circlearrow{Y}$. As a result, it holds that
\begin{align}
    \frac{H(K)}{n} \leq \underset{ \substack{U \\{\substack{U \circlearrow{X} \circlearrow{Y}\\ I(U;X)-I(U;Y) \leq C_{\eta}(P,N_{T} \times N_{R})+\alpha'(n)}}}}{\max} I(U;X).
    \nonumber
\end{align}
Here, $\underset{n\rightarrow \infty}{\lim}\alpha'(n)$ can be made arbitrarily small by choosing $\alpha,\epsilon$ to be arbitrarily small positive constants. This completes the converse proof of Theorem \ref{ccretathmMIMO}.
\section{Proof of Theorem \ref{achievratecompoundchannels}}
\label{prooflowerbound}
Let $a>0$ be fixed arbitrarily. Let $\mathcal{G}$ be an arbitrary closed subset of $\mathcal{G}_a$ defined in \eqref{setga}.
For any $\mbf g \in \mathcal{G}$, we consider the channel $W_{\mbf g}$:
\begin{align}
\bs{z}_{i}=\mbf g\bs{t}_{i}+\bs{\xi}_{i} \quad i=1, \hdots,n,  \nonumber
\end{align}
where $\bs{t}^n=(\bs{t}_1,\hdots,\bs{t}_n)\in\mbb C^{N_{T}\times n}$ and $\mbf \bs{z}^n=(\bs{z}_1,\hdots,\bs{z}_n)\in \mbb C^{N_{R}\times n}$ are channel input and output blocks, respectively.
Again, here, the block-length $n$ can vary and for the channel $W_\mbf g, \ \mbf g \in \mc G,$ the arguments determine the block-length.
The $\bs{\xi}_i$s are mutually independent complex Gaussian random vectors, each with mean $\mbf 0_{N_{R}}$ and covariance matrix $\sigma^2 \mbf I_{N_R},$ and independent of the random input sequence $\bs{T}^n.$
Let $\mc T^n$ and  $\mc Z^n$ be the set of all $\bs{t}^n$, $\bs{z}^n$, respectively. Let $q(\bs{z}^n)$ be an arbitrary output distribution for the channel $W_\mbf g,\ \mbf g \in \mc G.$
We define for any $\bs{t}^n\in \mc T^n$ and any $\bs{z}^n \in \mc Z^n$ 
$$i_{\mbf g}(\bs{t}^n, \bs{z}^n)=\log \frac{W_{\mbf g}(\bs{z}^n| \bs{t}^n) }{q( \bs{z}^n)}.$$
The goal is to prove that $\underset{\mbf Q \in \mathcal{Q}_{P}}{\sup}\underset{\mbf g \in \mathcal{G}}{\inf} \log\det(\mbf I_{N_{R}}+\frac{1}{\sigma^2}\mbf g \mbf Q \mbf g^H)$ is an achievable rate for  $\mathcal{C}=\{ W_{\mbf g}: \mbf g\in \mathcal{G}\}.$  
 Our proof is inspired by \cite{discretetimegaussian}\footnote{In\cite{discretetimegaussian}, the focus was on compound real Gaussian channels with square channel matrix whose operator norm is upper-bounded by $a$ and with noise covariance matrix satisfying further conditions. In our work, we consider \textit{compound complex Gaussian} channels with \textit{fixed noise covariance matrix equal to} $\sigma^2 \mbf I_{N_{R}}$ and with
channel matrix that has \textit{arbitrary dimension} and whose operator norm does not exceed $a.$}.  
\subsection{Proof of Theorem \ref{achievratecompoundchannels} for finite  subsets of $\mathcal{G}$}
We consider first finite subsets of $\mathcal{G}$ and prove the following theorem.
\begin{theorem}
\label{capacityfinitestate}
Let $\mathcal{G}^\prime$ be any finite subset of $\mathcal{G}.$
We define the compound channel 
$$ \mathcal{C}^\prime=\{ W_{\mbf g}: \mbf g\in \mathcal{G}^\prime\}.$$
An achievable rate for $\mathcal{C}'$  is 
$$ \underset{\mbf Q \in \mathcal{Q}_{P}}{\sup}\underset{\mbf g\in\mathcal{G}^{\prime}}{\inf} \log\det(\mathbf{I}_{N_{R}}+\frac{1}{\sigma^2}\mathbf{g}\mathbf{Q}\mathbf{g}^{H}).  $$
\end{theorem}
\begin{proof}
In order to prove Theorem \ref{capacityfinitestate}, we introduce the following lemmas first.
\begin{lemma}{(Feinstein’s Lemma with Input Constraints: Lemma 1 in \cite{discretetimegaussian}\color{black})}\\~\\
Consider any channel  $W$. The input set  is denoted by $\mc T$ and the output set is denoted by $\mc Z$. The random input and output are denoted by $T$ and $Z,$ respectively. Let $p(t)$ and $q(z)$ be arbitrary input and output distribution, respectively. We further define $i(T,Z)=\log\frac{W(Z|T)}{q(Z)}.$ 
Then, for any integer $\tau \geq 1$, real number $\alpha>0$, and measurable subset $E$ of $\mc T$, there exists a code with cardinality $\tau$, maximum error probability $\epsilon$ and block-length $n=1$, whose codewords are contained in the set $E,$ where $\epsilon$ satisfies
$$\epsilon=\tau2^{-\alpha}+\mbb P\left[ i(T,Z)\leq \alpha\right]+\mbb P\left[T \notin E \right],        $$
where
\begin{align}
    \mbb P\left[T \notin E\right]=\int_{E^c}p(t) dt. \nonumber
\end{align}
and
\begin{align}
\mbb P\left[ i(T,Z)\leq \alpha\right]=\int_{i(t,z)\leq \alpha} W(z|t)p(t) dt dz. \nonumber
\end{align}
\end{lemma}
\begin{proof}
The proof is the same as the one for Theorem 2 in \cite{errorbound} or Lemma 8.2.1 in \cite{informationbook}.
\end{proof}
Now, by applying Lemma 2 in \cite{discretetimegaussian} to the $n$-extension channel for $E_n=\{ \bs{t}^n=(\bs{t}_1,\hdots,\bs{t}_n)\in \mbb C^{N_T\times n}: \frac{1}{n}\sum_{i=1}^{n}\lVert \bs{t}_i \rVert^2\leq P\},$ we obtain the following lemma:
\begin{lemma}
\label{existenceerror}
Let $\mathcal{G}^{\prime}$ be a finite subset of $\mathcal{G}$. Let $\mathcal{C}'$ be the compound channel defined as
$$\mathcal{C}'=\{W_{\mbf g}: \mbf g \in \mathcal{G}^{\prime}      \}$$
and let $p(\bs{t})$ be an input probability density function determining $p_{\mbf g}(\bs{z},\bs{t})= W_{\mbf g}(\bs{z}|\bs{t})p(\bs{t})$ and  $i_{\mbf g}(\bs{t},\bs{z})=\log\frac{W_{\mbf g}(\bs{z}|\bs{t})}{q(\bs{z})}$ with $q(\bs{z})$ being an output probability density function.
We denote the random input and output sequence by $\bs{T}^n$ and $\bs{Z}^n,$ respectively. Then for any real numbers $\alpha>0$, $\delta>0$, and any integer $\tau\geq 1$, there exists a code $\Gamma_n$ for $\mathcal{C}'$ with size $|\Gamma|=\tau$, block-length $n$ and with codewords contained in $E_n=\{ \bs{t}^n=(\bs{t}_1,\hdots,\bs{t}_n)\in \mbb C^{N_T\times n}: \frac{1}{n}\sum_{i=1}^{n} \lVert \bs{t}_i \rVert^2\leq P\}$ such that for all $\mbf g \in \mc G'$
\begin{align}
    e(\Gamma_n,\mbf g) \leq |\mathcal{G}^{\prime}|\tau2^{-\alpha}+|\mathcal{G}^{\prime}|^22^{-\delta}+|\mathcal{G}^{\prime}|\mbb P\left[\bs{T}^n\notin E_n\right]+\sum_{\mbf g \in \mathcal{G}^{\prime}}\mbb P\left[i_{\mbf g}(\bs{T}^n,\bs{Z}^n)\leq \alpha+\delta\right]. \nonumber
\end{align}
\end{lemma}
\begin{proof}
The proof is a simple modification of that of Lemma 3 in \cite{capacityofclassofchannels}. It is based on an application of Feinstein's lemma.
\end{proof}

\begin{lemma}
\label{abweichungmean} 
Let $W_{\mbf g}$ be a fixed channel with $\mbf g \in \mathbb{C}^{N_R\times N_T}$. Let $\bs{T}^n \in \mbb C^{N_T\times n}$ and $\bs{Z}^n \in \mbb C^{N_R\times n}$ be the random input and output sequence, respectively. We further assume that the $\bs{T}_is, i=1,\hdots,n$ are i.i.d., where each $\bs{T}_i \in \mbb C^{N_T}$ is Gaussian distributed with mean zero and with a non-singular covariance matrix $\mbf Q.$
Then for any $\delta>0$
\begin{align} 
    \mbb P\left[i_{\mbf g}(\bs{T}^n,\bs{Z}^n)\leq \mbb E\left[i_{\mbf g}(\bs{T}^n,\bs{Z}^n)\right]-n\delta\right] \leq 2^{\left\{-\frac{n N_R}{2\ln(2)}\left[\left(1+\frac{(\ln(2)\delta)^2}{N_{R}^2}\right)^\frac{1}{2}-1\right] \right\}}.  \nonumber 
\end{align}
\end{lemma}
\begin{proof}
Since $(\bs{T}_i,\bs{Z}_i), i=1, \hdots,n,$ are i.i.d., we introduce $(\bs{T},\bs{Z})$ such that $(\bs{T},\bs{Z})$ has the same joint distribution as each of the $(\bs{T}_i,\bs{Z}_i).$
Now 
\begin{align}
    \mbb E[i_{\mbf g}(\bs{T}^n, \bs{Z}^n)]&=n \mbb E\left[i_{\mbf g}(\bs{T},\bs{Z})\right] \nonumber \\
    &=n \log\det\left(\mbf I_{N_{R}}+\frac{1}{\sigma^2}\mbf g\mbf Q \mbf g^H\right). \nonumber
\end{align}
Let
\begin{align}
    \mbf \Theta=\mbf g \mbf Q \mbf g^{H}+\sigma^2\mbf I_{N_R} \nonumber
\end{align}
be the covariance matrix of $\bs{Z}.$
Here, $\mbf \Theta$ is positive definite and therefore non-singular since $\mbf Q$ is positive definite (non-singular covariance matrix). We further define
\begin{align}
   \phi_i= -\frac{1}{\sigma^2}\left(\bs{Z}_i-\mbf g \bs{T}_i\right)^H\left(\bs{Z}_i-\mbf g \bs{T}_i\right)+\bs{Z}_i^H\mbf \Theta^{-1}\bs{Z}_i.
    \nonumber 
\end{align}
Since the $\bs{\phi}_is$ are i.i.d., we define $\bs{\phi}$ to be a random variable with the same distribution as each of the $\bs{\phi}_i$ as follows:
\begin{align}
    \phi=-\frac{1}{\sigma^2}\left(\bs{Z}-\mbf g \bs{T}\right)^H\left(\bs{Z}-\mbf g \bs{T}\right)+\bs{Z}^H\mbf \Theta^{-1}\bs{Z}. \nonumber
\end{align}
Since  $ i_{\mbf g}(\bs{T}_i,\bs{Z}_i)=\log\det\left(\mbf I_{N_{R}}+\frac{1}{\sigma^2}\mbf g\mbf Q \mbf g^H\right)+\frac{\bs{\phi}_i}{\ln(2)}, i=1,\hdots,n,$ it follows that
\begin{align}
    \mbb P\left[i_{\mbf g}(\bs{T}^n,\bs{Z}^n)\leq \mbb E\left[i_{\mbf g}(\bs{T}^n,\bs{Z}^n )\right]-n\delta\right] 
   &=\mbb P\left[ \sum_{i=1}^{n}\frac{\phi_i}{\ln(2)} \leq -n\delta\right] \nonumber \\
   &=\mbb P\left[ -(\ln(2)n\delta+\sum_{i=1}^{n}\phi_i)\geq 0 \right] \nonumber \\
   &\leq \mbb E\left[\exp(-\beta(\ln(2)n\delta+\sum_{i=1}^{n}\bs{\phi}_i))\right] \nonumber \\
   &=\exp(-\beta n \ln(2) \delta)\mbb E\left[\exp(-n\beta\bs{\phi})\right]\quad \forall \beta\geq 0, \nonumber
\end{align}
where we used the Chernoff's bound.
Let $\zeta(\beta)=\mbb  E\left[ \exp(-\beta\phi)   \right]$ so that
\begin{align}
    \mbb P\left[i_{\mbf g}(\bs{T}^n,\bs{Z}^n)\leq \mbb E\left[i_{\mbf g}(\bs{T}^n,\bs{Z}^n)\right]-n\delta\right]\leq \left(\exp(-\ln(2)\beta\delta)\zeta(\beta)\right)^n.
    \label{probj}
\end{align}
In order to compute $\zeta(\beta),$ we introduce the Gaussian random vector $\bs{W}=[\bs{T},\bs{Z}]^T$ of dimension $N_{T}+N_{R}.$ Since $\bs{T}$ and $\bs{Z}$ have mean zero, $\bs{W}$ has also mean zero and its covariance matrix $\mbf O $ can be written as:
\begin{align}\mbf O=
\begin{pmatrix}
 \begin{matrix}
  \mbf Q & \mbf Q \mbf g^{H} \\
  \mbf g \mbf Q & \mbf \Theta
  \end{matrix}
  \end{pmatrix}. \nonumber
\end{align}
The block-matrix $\mbf O$ is non-singular, since $\mbf Q$ is non-singular and  the matrix $\mbf\Theta-\mbf g \mbf Q \mbf Q^{-1}\mbf Q\mbf g^{H}=\sigma^2\mbf I_{N_R}$ is non-singular.
We further define:
\begin{align}
    \mbf \Lambda=
    \begin{pmatrix}
    \begin{matrix}
    \mbf 0 & \mbf 0 \\
  \mbf 0 & \mbf \Theta^{-1}
    \end{matrix}
    \end{pmatrix} \nonumber
\end{align}
and 
\begin{align}
    \mbf \Phi=
    \begin{pmatrix}
    \begin{matrix}
      \frac{1}{\sigma^2}\mbf g^{H}\mbf g & -\frac{1}{\sigma^2}\mbf g^{H} \\
  -\frac{1}{\sigma^2}\mbf g & \frac{1}{\sigma^2}\mbf I_{N_{R}}
    \end{matrix}
    \end{pmatrix}. \nonumber
\end{align}
We can then write
\begin{align}
   \bs{\phi}=\bs{W}^{H}\bs{\Lambda}\bs{W}-\bs{W}^{H}\mbf \Phi\bs{W}. \nonumber
\end{align}
Indeed
\begin{align}
    \bs{W}^{H}\mbf \Lambda\bs{W}&=(\bs{T}^H\bs{Z}^H)      \begin{pmatrix}
    \begin{matrix}
    \mbf 0 & \mbf 0 \\
  \mbf 0 & \mbf \Theta^{-1}
    \end{matrix}
    \end{pmatrix}\begin{pmatrix}\begin{matrix}
      \bs{T} \\
  \bs{Z}
    \end{matrix}\end{pmatrix} \nonumber \\
    &=(\mbf 0 \ \bs{Z}^H\mbf \Theta^{-1})\begin{pmatrix}\begin{matrix}
      \bs{T} \\
  \bs{Z}
    \end{matrix}\end{pmatrix} \nonumber \\
    &=\bs{Z}^H\mbf \Theta^{-1}\bs{Z}, \nonumber
\end{align}
and
\begin{align}
    \bs{W}^{H}\mbf \Phi\bs{W}=
    &(\bs{T}^H\bs{Z}^H)        \begin{pmatrix}
    \begin{matrix}
      \frac{1}{\sigma^2}\mbf g^{H}\mbf g & -\frac{1}{\sigma^2}\mbf g^{H} \\
  -\frac{1}{\sigma^2}\mbf g & \frac{1}{\sigma^2}\mbf I_{N_{R}}
    \end{matrix}
    \end{pmatrix}\begin{pmatrix}\begin{matrix}
      \bs{T} \\
  \bs{Z}
    \end{matrix}\end{pmatrix}  \nonumber \\
    &=(\frac{1}{\sigma^2} \bs{T}^{H} \mbf g^{H} \mbf g -\frac{1}{\sigma^2}\bs{Z}^{H}\mbf g  \quad -\frac{1}{\sigma^2}\bs{T}^{H}\mbf g^{H}+\frac{1}{\sigma^2}\bs{Z}^{H})\begin{pmatrix}\begin{matrix}
      \bs{T} \\
  \bs{Z}
    \end{matrix}\end{pmatrix} \nonumber \\
    &=\frac{1}{\sigma^2}\left[\bs{T}^{H}\mbf g^{H} \mbf g \bs{T}-\bs{Z}^{H}\mbf g \bs{T}-\bs{T}^{H}\mbf g^{H} \bs{Z}+\bs{Z}^{H}\bs{Z} \right] \nonumber \\
    &=\frac{1}{\sigma^2}(\bs{Z}-\mbf g \bs{T})^{H}(\bs{Z}-\mbf g \bs{T}). \nonumber
\end{align}
Now, it follows that
\begin{align}
    \zeta(\beta)&= \mbb E\left[\exp(-\beta\bs{\phi})  \right]\nonumber \\
    &=\frac{\int\exp(-\bs{w}^H\mbf O^{-1}\bs{w})\times
    \exp\left[-\beta(\bs{w}^H\bs{\Lambda}\bs{w}-\bs{w}^H\bs{\Phi}\bs{w})\right]d\bs{w}}{\pi^{N_{T}+N_{R}}\det(\mbf O)},  \nonumber
\end{align}
where the integral is a $(N_{T}+N_{R})$-fold integral over $\mbb C^{N_{T}+N_{R}}.$
Let $\mbf  M(\beta)=\mbf O^{-1}+\beta(\Lambda-\Phi)  \in \mbb C^{(N_{T}+N_{R})\times (N_{T}+N_{R})}.$ Here, $\mbf M(\beta)$ is positive definite  for $\beta<\beta_0$ for some $\beta_0\geq 1.$
To prove this, we define
 $\mbf O_1=\mbf Q, \mbf O_2=\mbf Q\mbf g^{H}, \mbf O_3=\mbf g \mbf Q, \ \text{and} \ \mbf O_4=\mbf \Theta. $
Since $\mbf O_1$ is invertible and $\mbf O_4-\mbf O_3\mbf O_1^{-1}\mbf O_2=\sigma^2 \mbf I_{N_{R}}$ is also invertible, it follows by applying the inversion rule for the block-matrix $\mbf O$  that
\begin{align}\mbf O^{-1}=
\begin{pmatrix}
 \begin{matrix}
  \mbf Q^{-1}+\frac{1}{\sigma^2}\mbf g^{H}\mbf g & -\frac{1}{\sigma^2}\mbf g^{H} \\
  -\frac{1}{\sigma^2}\mbf g  & \frac{1}{\sigma^2}\mbf I_{N_{R}}
  \end{matrix}
  \end{pmatrix}. \nonumber
\end{align}
Now
\begin{align}\mbf M(\beta)=
\begin{pmatrix}
 \begin{matrix}
  \mbf Q^{-1}+\frac{1}{\sigma^2}(1-\beta)\mbf g^{H}\mbf g & -\frac{1}{\sigma^2}(1-\beta)\mbf g^{H} \\
  -\frac{1}{\sigma^2}(1-\beta)\mbf g  & \beta\mbf \Theta^{-1}+\frac{1}{\sigma^2}(1-\beta)\mbf I_{N_{R}}
  \end{matrix}
  \end{pmatrix} \quad \beta\geq 0.
  \nonumber \end{align}
  Notice that for $\beta\geq 0$ 
  $$ M(\beta)=\beta M(1)+(1-\beta)M(0) $$
  and that $\mbf M(0)$ and $\mbf M(1)$ are both positive definite. From the convexity of the set of positive definite Hermitian matrices, it follows for all $\beta\in(0,1)$ that
  $\beta \mbf M(1)+(1-\beta)\mbf M(0)$ is positive definite. This proves that $\mbf M(\beta)$ is positive definite for $\beta<\beta_0$ for some $\beta_0\geq 1.$

Now, we have
\begin{align}
    \zeta(\beta)&=\frac{1}{\pi^{N_{T}+N_{R}}\det(\mbf O)}\int \exp(-\bs{w}^H\mbf M(\beta) \bs{w}) d\bs{w} \nonumber \\
    &=\det(\mbf M(\beta) \mbf O)^{-1}. \nonumber
\end{align}
\color{black}Indeed, we write the SVD decomposition of the Hermitian positive definite matrix as $\mbf M(\beta)=\mbf U \mbf D \mbf U^{H}$, where $\mbf D$ contains the eigenvalues $\kappa_1,\hdots,\kappa_{N_{T}+N_{R}}$ of $\mbf M(\beta).$
It follows for $\bs{w}=(w_1,\hdots,w_{N_T+N_R})^{T}\in \mbb C^{N_T+N_R}$ that
\begin{align}
    \langle \bs{w},\mbf M(\beta)\bs{w} \rangle&= \langle \mbf U^{H}\bs{w},\mbf D \mbf U^{H}\bs{w}\rangle \nonumber \\
    &=\sum_{i=1}^{N_{T}+N_{R}}\kappa_i |\bs{w}_i|^2
    \nonumber
\end{align}
and that
\begin{align}
\int \exp(-\bs{w}^H\mbf M(\beta) \bs{w}) d\bs{w}&=\prod_{i=1}^{N_T+N_R}\int\exp\left(-\kappa_i|\bs{w}_i|^{2} \right) d\bs{w}_i  \nonumber \\
&= \prod_{i=1}^{N_T+N_R} \pi \frac{1}{\kappa_i} \nonumber \\
&=\pi^{N_T+N_R}\frac{1}{\det(\mbf M(\beta))}. \nonumber
\end{align}
\color{black}
By substituting $\mbf \Lambda,$ $\mbf \Phi$ and $\mbf O$, and by using the fact that
$\mbf \Theta=\mbf g \mbf Q \mbf g^H+\sigma^2 \mbf I_{N_{R}},$ we obtain
\begin{align}
        \mbf M(\beta)\mbf O
        &= \mbf I_{N_{T}+ N_{R}}+\beta(\mbf \Lambda-\mbf \Phi)\mbf O \nonumber \\
        &= \mbf I_{N_{T}+ N_{R}}+\beta  \begin{pmatrix}
    \begin{matrix}
       \frac{-1}{\sigma^2}\mbf g^{H} \mbf g& \frac{1}{\sigma^2}\mbf g^{H} \\
  \frac{1}{\sigma^2}\mbf g&\mbf \Theta^{-1}-\frac{1}{\sigma^2}\mbf I_{N_{R}}
    \end{matrix}
    \end{pmatrix}\begin{pmatrix}
 \begin{matrix}
  \mbf Q & \mbf Q \mbf g^{H} \\
  \mbf g \mbf Q & \mbf \Theta
  \end{matrix}
  \end{pmatrix}
  \nonumber \\
  &=\mbf I_{N_{T}+ N_{R}} +\beta \begin{pmatrix}
    \begin{matrix}
       \frac{-1}{\sigma^2}\mbf g^{H} \mbf g \mbf Q+\frac{1}{\sigma^2}\mbf g^{H}\mbf g \mbf Q& -\frac{1}{\sigma^2}\mbf g^{H}\mbf g \mbf Q \mbf g^{H}+\frac{1}{\sigma^2}\mbf g^{H}\mbf \Theta \\
  \frac{1}{\sigma^2}\mbf g\mbf Q+\mbf \Theta^{-1}\mbf g\mbf Q-\frac{1}{\sigma^2}\mbf g \bf Q&\frac{1}{\sigma^2}\mbf g \mbf Q \mbf g^{H}+\mbf I_{N_{R}}-\frac{1}{\sigma^2}\mbf \Theta
    \end{matrix}
    \end{pmatrix} \nonumber \\
    &=\begin{pmatrix}
    \begin{matrix}
      \mbf I_{N_{T}} & \beta \mbf g^{H} \\
  \beta \mbf \Theta^{-1} \mbf g \mbf Q & \mbf I_{N_{R}}
    \end{matrix}
    \end{pmatrix}. \nonumber
\end{align}
As a result, we obtain using the determinant rule for block-matrices
\begin{align}
    \det(\mbf M(\beta) \mbf O)
    &=  \det(\mbf I_{N_{R}}-\beta^2 \mbf \Theta^{-1}\mbf g \mbf Q \mbf g^{H}) \nonumber \\
    &=\det(\mbf \Theta^{-1})\det(\mbf \Theta-\beta^2 \mbf g \mbf Q \mbf g^{H}) \nonumber \\
    &=\det(\mbf \Theta^{-1})\det(\sigma^2\mbf I_{N_{R}}+(1-\beta^2) \mbf g \mbf Q \mbf g^{H}) \nonumber \\
    &=\sigma^{2N_{R}}\frac{\det(\mbf I_{N_{R}}+(1-\beta^2) \frac{1}{\sigma^2}\mbf g \mbf Q \mbf g^{H})}{\det(\mbf \Theta)}, \nonumber 
\end{align}
where
\begin{align}
    \det(\mbf \Theta)=\sigma^{2N_{R}}\det(\mbf I_{N_{R}}+ \frac{1}{\sigma^2}\mbf g \mbf Q \mbf g^{H}). \nonumber
\end{align}
We define $\lambda_1,\hdots, \lambda_{N_{R}}$ to be the eigenvalues of the positive definite matrix $\frac{1}{\sigma^2}\mbf g \mbf Q \mbf g^{H}.$
Then it holds that
\begin{align}
    \det(\mbf I_{N_{R}}+ \frac{1}{\sigma^2}\mbf g \mbf Q \mbf g^{H})=\prod_{i=1}^{N_{R}} (1+\lambda_i) \nonumber
\end{align}
and
\begin{align}
    \det(\mbf I_{N_{R}}+(1-\beta^2) \frac{1}{\sigma^2}\mbf g \mbf Q \mbf g^{H}) =\prod_{i=1}^{N_{R}} (1+(1-\beta^2)\lambda_i). \nonumber
\end{align}
This yields
\begin{align}
    \det(\mbf M(\beta)\mbf O)=\prod_{i=1}^{N_{R}} \frac{1+(1-\beta^2)\lambda_i}{1+\lambda_i}
    &=\prod_{i=1}^{N_{R}}\left(1-\beta^2\frac{\lambda_i}{1+\lambda_i}\right) \nonumber
\end{align}
such that 
\begin{align}
    \zeta(\beta)= \prod_{i=1}^{N_{R}}\left(1-\beta^2\frac{\lambda_i}{1+\lambda_i}\right)^{-1} \quad 0\leq \beta <\beta_0. \nonumber
\end{align}
Then, we have
\begin{align}
    \zeta(\beta)\leq \frac{1}{(1-\beta^2)^{N_{R}}} \quad 0 \leq \beta <\beta_0 \nonumber
\end{align}
and hence
\begin{align}
    \left(\exp(-\ln(2)\delta\beta)\zeta(\beta)\right)^{\frac{1}{N_R}} \leq \frac{\exp(-\frac{\ln(2)\delta\beta}{N_R})}{1-\beta^2} \quad 0\leq \beta <\beta_0. \nonumber
\end{align}
Now if we put
\begin{align}
    \beta=\frac{N_R}{\ln(2)\delta}\left[ -1+\left(1+\frac{(\ln(2)\delta)^2}{N_R^2} \right)^{\frac{1}{2}}\right], \nonumber
\end{align}
it follows that $0<\beta<1$ and it holds that
\begin{align}
    \exp(-\frac{\ln(2)\delta\beta}{N_{R}})=\exp\left(-´\left[-1+\left(1+\frac{(\ln(2)\delta)^2}{N_R^2}     \right)^{1/2}     \right] \right)   \nonumber
\end{align}
and that
\begin{align}
\frac{1}{1-\beta^2}&=\frac{1}{1-\left(\frac{N_R}{\ln(2)\delta}\right)^2\left[1-2\sqrt{1+\left(\frac{\ln(2)\delta}{N_R}\right)^2}+1+\left(\frac{\ln(2)\delta}{N_R}\right)^2\right]} \nonumber \\
&=\frac{1}{2}\frac{\left(\frac{\ln(2)\delta}{N_R}\right)^2}{\sqrt{1+\left(\frac{\ln(2)\delta}{N_R}\right)^2}-1} \nonumber \\
&=\frac{1}{2}\left(\sqrt{1+\left(\frac{\ln(2)\delta}{N_{R}}\right)^2}+1\right) \nonumber \\
&=\left(1+\frac{1}{2}\left[-1+\left(1+\frac{(  \ln(2)\delta)^2}{N_R^2}     \right)^{1/2}     \right]        \right). \nonumber
\end{align}
This implies that
\begin{align}
    (1-\beta^2)^{-1} \exp(-\frac{\ln(2)\delta\beta}{N_R})=\left(1+\frac{1}{2}\left[-1+\left(1+\frac{(\ln(2)\delta)^2}{N_{R}^2}     \right)^{1/2}     \right]        \right) \exp\left(-\left[-1+\left(1+\frac{(\ln(2)\delta)^2}{N_{R}^2}     \right)^{1/2}     \right]  \right).    \nonumber 
\end{align}
Since $(1+\frac{1}{2}x)\exp(-x)\leq \exp(-\frac{x}{2}) \ \text{for} \ x\geq 0,$ we have
\begin{align}
\exp(-\ln(2)\delta\beta) \zeta(\beta) \leq \exp\left(-\frac{N_R}{2}\left[\left(1+\frac{(\ln(2)\delta)^2}{N_R^2}     \right)^{1/2}-1     \right]\right). \nonumber
\end{align}
It follows from \eqref{probj} that
\begin{align}
   \mbb P\left[i_{\mbf g}(\bs{T}^n,\bs{Z}^n)\leq \mbb E\left[i_{\mbf g}(\bs{T}^n,\bs{Z}^n)\right]-n\delta\right] &\leq \exp\left(-\frac{nN_R}{2}\left[\left(1+\frac{(\ln(2)\delta)^2}{N_R^2}     \right)^{1/2}-1     \right]\right) \nonumber \\
   &= 2^{\left\{-\frac{n N_R}{2\ln(2)}\left[\left(1+\frac{(\ln(2)\delta)^2}{N_{R}^2}\right)^\frac{1}{2}-1\right] \right\}}. \nonumber 
\end{align}
This completes the proof of the lemma.
\end{proof}
\begin{lemma}
\label{upperboundprobsum}
Let $\bs{X}_i,$ $i=1, \hdots,n$ be i.i.d. $N$-dimensional complex Gaussian random vectors with mean zero and covariance matrix $\mbf O$ whose trace is smaller than or equal to $M$.
Then, for any $\delta>0$
\begin{align}
    \mbb P\left[\sum_{i=1}^{n} \lVert \bs{X}_i\rVert^2\geq n(M+\delta)\right]\leq \left[(1+\frac{\delta}{M})2^{-\frac{\delta}{\ln(2)M}}\right]^{n}, \nonumber
\end{align}
where 
\begin{align}
    \lVert \bs{X}_i \rVert^2=\sum_{j=1}^{N} |\bs{X}_i^j|^2 \nonumber
\end{align}
and
\begin{align}
    \bs{X}_i=(\bs{X}_i^1,\hdots,\bs{X}_i^N)^{T}. \nonumber
\end{align}
\end{lemma}
\begin{proof}
Let $\bs{X}$ be a random vector with the same distribution as each of the $\bs{X}_i$. Then
\begin{align}
\mbb P\left[ \sum_{i=1}^{n} \lVert \bs{X}_i\rVert^2 \geq n(M+\delta)      \right] 
&=\mbb P\left[ \sum_{i=1}^{n}\lVert \bs{X}_i^2\rVert-n(M+\delta)\geq 0       \right] \nonumber \\
    &\leq \mbb E \left[\exp\left(\beta\left(\sum_{i=1}^{n}\lVert \bs{X}_{i}\rVert^2-n(M+\delta\right)    \right)            \right]\nonumber \\
    &=\left[\exp(-[M+\delta]\beta)\mbb E\left[ \exp(\beta \lVert \bs{X} \rVert^2              \right]\right]^n,
    \label{eqprob}
\end{align}
where we used the $\bs{X}_is$ are i.i.d..
By a standard calculation which follows below, one can show that
\begin{align}
    \mbb E\left[ \exp(\beta \lVert \bs{X}\rVert^2    \right]&=\mbb E \left[\exp(\beta\bs{X}^H\bs{X})           \right] \nonumber \\
    &=\prod_{j=1}^{N}(1-\beta\mu_j)^{-1} \quad \beta<\beta_0, \nonumber
\end{align}
where $\mu_1,\hdots,\mu_{N}$ are eigenvalues of $\mbf O$, and for $\beta_0=\frac{1}{M}\leq \frac{1}{\mu_1+\hdots+\mu_N}<\underset{j\in \{1,\hdots,N\}}{\min}\frac{1}{\mu_j}$  so that all the factors are positive, whether $\mbf O$ is non-singular or singular.
To prove this, we let $r$ be the rank of $\mbf O.$ It holds that $r\leq N$. We make use of the spectral decomposition theorem to express $\mbf O$ as $\mbf S_{\mbf O}^{\star} \Lambda^{\star}{S_{\mbf O}^{\star}}^{H} $, where $\Lambda^{\star}$ is a diagonal matrix whose first $r$ diagonal elements are positive and where the remaining diagonal elements are equal to zero.
Next, we let $\mbf V^{\star}=\mbf S_{\mbf O}^{\star} {\Lambda^{\star}}^{\frac{1}{2}}$ and remove the $N-r$ last columns of $\mbf V^{\star}$, which are null vectors to obtain the matrix $\mbf V.$ Then, it can be verified that $\mbf O=\mbf V \mbf V^{H}.$
We can write $\bs{X}=\mbf  V \bs{U}^\star$
where $\bs{U}^\star \sim\mathcal{N}_{\mbb C}(\bs{0},\mbf I_{r}).$
As a result:
\begin{align}
    \bs{X}^H\bs{X}={(\bs{U}^\star)}^{H}\mbf V^{H}\mbf V \bs{U}^\star. \nonumber
\end{align}
Let $\mbf S$ be a unitary matrix which diagonalizes $\mbf V^{H} \mbf V$ such that $\mbf S^{H} \mbf V^{H} \mbf V \mbf S= \text{Diag}(\mu_1,\hdots,\mu_r)$ with $\mu_1,\hdots,\mu_r$ being the positive eigenvalues of $\mbf O=\mbf V \mbf V^{H}$ in decreasing order.
One defines $\bs{U}=\mbf S^{H} \bs{U}^{\star}.$ We have
\begin{align}
    \text{cov}(\bs{U})&=\mbf S^{H} \text{cov}(\bs{U}^{\star}) \mbf S \nonumber \\
    &=\mbf S^{H} \mbf S \nonumber \\
    &=\mbf I_{r}. \nonumber
\end{align}
Therefore, it holds that $\bs{U} \sim \mathcal{N}(\bs{0},\mbf I_r).$
Since $\mbf S$ is unitary, it follows that
\begin{align}
    \bs{X}^{H}\bs{X}&=\left((\mbf S^{H})^{-1}\bs{U}\right)^{H} \mbf V^{H} \mbf V (\mbf S^{H})^{-1} \bs{U} \nonumber \\
    &=\bs{U}^{H}\mbf S^{H} \mbf V^{H} \mbf V \mbf S \bs{U} \nonumber \\
    &=\bs{U}^{H}\text{Diag}(\mu_1,\hdots,\mu_r)\bs{U} \nonumber \\
    &=\sum_{j=1}^{r} \mu_j |\bs{U}_{j}|^{2}.\nonumber
\end{align}
Then, we have
\begin{align}
    \mbb E\left[ \exp(\beta \lVert \bs{X}\rVert^2)    \right] &= \mbb E \left[\prod_{j=1}^{r}\exp(\frac{1}{2}\beta\mu_j 2|\bs{U}_j|^2  )             \right] \nonumber \\
    \nonumber \\
    &=\prod_{j=1}^{r} \mbb E \left[ \exp(\frac{1}{2}\beta\mu_j 2|\bs{U}_j|^2  )       \right] \nonumber \\
    &=\prod_{j=1}^{N} (1-\beta\mu_j)^{-1},             \nonumber 
\end{align}
where we used that all the $\bs{U}_j$'s are independent, that $ \forall j \in \{1,\hdots,r\},  2|\bs{U}_{j}|^2$ is chi-square distributed with $k=2$ degrees of freedom and with moment generating function equal to $\mbb E\left[\exp(2t|\bs{U}_j|^2)\right]=(1-2t)^{-k/2}$ for $t<\frac{1}{2}$ and that $\forall j \in \{1,\hdots,r\}$ and for $\beta<\beta_0,$ $\frac{1}{2}\beta\mu_j<\frac{1}{2}$. This completes the standard calculation.

Now, it holds that
$$ \prod_{i=1}^{N}(1-\beta\mu_i)\geq 1-\beta(\mu_1+\hdots+\mu_{N})\geq 1-\beta M.         $$ This yields
\begin{align}
    &\exp(-(M+\delta)\beta)\mbb E\left[ \exp(\beta \lVert \bs{X}\rVert^2    \right] \leq \frac{\exp(-(M+\delta)\beta)}{1-\beta M}, \nonumber
\end{align}
where $0<\beta<\frac{1}{M}=\beta_0.$
Putting $\beta=\frac{\delta}{M(\delta+M)}<\frac{1}{M}$ yields
\begin{align}
    \exp(-(M+\delta)\beta)\mbb E\left[ \exp(\beta \lVert \bs{X}\rVert^2    \right]&\leq (1+\frac{\delta}{M})\exp(-\frac{\delta}{M})\nonumber \\
    &=(1+\frac{\delta}{M})2^{(-\frac{\delta}{\ln(2)M})},
\nonumber \end{align}
which combined with \eqref{eqprob} proves the lemma.
\end{proof}
\begin{lemma}
\label{existence} 
Let $\epsilon>0$ be fixed arbitrarily. Then, there exists a non-singular $\mbf Q \in \mathcal{Q}_{P}$ such that
\begin{enumerate}
    \item $\mathrm{tr}(\mbf Q)<P$
    \item$\log\det(\mbf I_{N_R}+\mbf g \mbf Q \mbf g^H)
    \geq \underset{\mbf Q \in \mathcal{Q}_{P}}{\sup}\underset{\mbf g \in \mathcal{G}}{\inf}\log\det(\mbf I_{N_{R}}+\frac{1}{\sigma^2}\mbf g \mbf Q \mbf g^H)-\epsilon \quad \text{for all} \ \mbf g \in \mathcal{G}$
    with $\mc G$ being any closed subset of $\mc G_a.$
\end{enumerate}
\end{lemma}
\begin{proof}
Notice first that the set  $\mathcal{L}=\mc G \times \mathcal{Q}_P$ is a  compact set, because the conditions on the matrices $\mbf g \in \mathcal{G}\subset \mc G_a$, on the positive semi-definite $\mbf Q$ guarantee that $\mc L$ is bounded and closed in $\mbb C^{N_{R}\times N_{T}} \times \mbb C^{N_{T}\times N_{T}}$. Now the function $f(\mbf g,\mbf Q)=\log\det(\mbf I_{N_{R}}+\frac{1}{\sigma^2}\mbf g \mbf Q \mbf  g^H)$ is uniformly continuous on $\mathcal{L}$ \color{black}. One can find a non-singular $\mbf Q_{0} \in \mathcal{Q}_P$ such that
\begin{align}
& \log\det(I_{N_R}+\frac{1}{\sigma^2}\mbf g \mbf Q_{0} \mbf g^H) \geq \underset{\mbf Q\in \mathcal{Q}_P}{\sup}\underset{\mbf g \in \mathcal{G}}{\inf}  \log\det(I_{N_R}+\frac{1}{\sigma^2}\mbf g \mbf Q \mbf g^H) -\frac{\epsilon}{2} \quad \forall \mbf g \in \mathcal{G}.\nonumber
\end{align}
If $\text{tr}(\mbf Q_0)<P$, there is nothing to prove. If $\text{tr}(\mbf Q_0)=P$, we can find, by the uniform continuity of $f$ on $\mathcal{L}$, a number $\delta>0$ such that
$\lvert f(\mbf g,\mbf Q)-f(\mbf g,\mbf Q_0)\rvert\leq \frac{\epsilon}{2}$ for all $\mbf g$ if $\lVert  \mbf Q- \mbf Q_0 \rVert \leq \delta$. We can then change $\mbf Q_0$ into a non-singular $\mbf Q_1$ in such a way that $\lVert \mbf Q_{1}-\mbf Q_{0} \rVert \leq \delta$ and
$\text{tr}(\mbf Q_{1})<\text{tr}(\mbf Q_{0})=P$.
 $\mbf Q_{1}$ satisfies the conditions of the lemma. This completes the proof of the lemma.
\end{proof}
Now that we proved the lemmas, we fix $R$ to be any positive number less than $\underset{\mbf Q \in \mathcal{Q}_{P}}{\sup}\underset{\mbf g\in\mathcal{G}^{\prime}}{\inf} \log\det(\mathbf{I}_{N_{R}}+\frac{1}{\sigma^2}\mathbf{g}\mathbf{Q}\mathbf{g}^{H})$ and put $2\theta=\underset{\mbf Q \in \mathcal{Q}_{P}}{\sup}\underset{\mbf g\in\mathcal{G}^{\prime}}{\inf} \log\det(\mathbf{I}_{N_{R}}+\frac{1}{\sigma^2}\mathbf{g}\mathbf{Q}\mathbf{g}^{H})-R.$\\~\\
By Lemma \ref{existence}, \color{black} one can find a non-singular $\mbf Q_1 \in \mathcal{Q}_{P}$ \color{black} such that $\text{tr}(\mbf Q_{1})=P-\beta, \quad \beta>0$ and
\begin{align}
\mbb E[i_{\mbf g}(\bs{T},\bs{Z})]=\log\det(\mbf I_{N_{R}}+\frac{1}{\sigma^2}\mbf g \mbf Q_{1} \mbf g^{H})\geq R+\frac{\theta}{2} \quad \forall \mbf g \in \mathcal{G}^{\prime},
\label{meanJ}
\end{align}
where $\bs{T}$ and $\bs{Z}$ represent the random input and output of $W_\mbf g,$ respectively.
 Let $E_n$ be the set of all input sequences $\bs{t}^n$ satisfying $\sum_{i=1}^{n}\lVert \bs{t}_i\rVert^2\leq nP.$   $ \text{For any} \ \mbf g \in \mc G',$  we define 
$\bs{T}_i,i=1, \hdots,n,$ be the i.i.d. random inputs of $ W_\mbf g,$ each with mean $\mbf 0_{N_T}$ and covariance matrix $\mbf Q_1$. Let $\hat{P}=P-\beta$ and $\hat{\beta}=\frac{\beta}{\ln(2)\hat{P}}-\log(1+\frac{\beta}{\hat{P}})>0$.
Then, by Lemma \ref{upperboundprobsum}, it holds that
\begin{align}
    \mbb P \left[\sum_{i=1}^{n}\lVert \bs{T}_i\rVert^2 \geq n(\hat{P}+\beta)  \right] &\leq \left[(1+\frac{\beta}{\hat{P}})2^{(-\frac{\beta}{\ln(2)\hat{P}})}\right]^{n} \nonumber \\
    &=2^{\left(-n\frac{\beta}{\ln(2)\hat{P}}+n\log(1+\frac{\beta}{\hat{P}})\right)} \nonumber \\
    &= 2^{-n\hat{\beta}}.\nonumber
\end{align}

As a result, we have
\begin{align}
    \mbb P\left[ \bs{T}^n \notin E_n\right]&= \mbb P \left[\sum_{i=1}^{n}\lVert \bs{T}_i\rVert^2 > nP \right]\nonumber \\
    &\leq \mbb P \left[ \sum_{i=1}^{n}\lVert \bs{T}_i\rVert^2 \geq  nP \right]\nonumber \\
    &=  \mbb P \left[\sum_{i=1}^{n}\lVert \bs{T}_i\rVert^2 \geq n(\hat{P}+\beta)  \right] \nonumber \\
    &\leq 2^{-n\hat{\beta}}.\nonumber
\end{align}

Now define $\tau=\lfloor2^{nR}\rfloor$, $\alpha=n(R+\frac{\theta}{8})$ and $\delta=\frac{n\theta}{8}.$
It follows from Lemma \ref{existenceerror} that there exists a code $\Gamma_n$  for $\mathcal{C}'$ with size $\lvert \Gamma_n \rvert=\tau$ and block-length $n$ such that for all $\mbf g \in \mc G'$ 
\begin{align}
    e(\Gamma_n,\mbf g)
    &\leq |\mathcal{G}^{\prime}|2^{nR}2^{-n(R+\frac{\theta}{8})}+|\mathcal{G}^{\prime}|^22^{-n\frac{\theta}{8}}+|\mathcal{G}^{\prime}|2^{-n\hat{\beta}}+\sum_{\mbf g\in\mathcal{G}^{\prime}} \mbb P\left[i_{\mbf g}(\bs{T}^n,\bs{Z}^n)\leq n(R+\frac{\theta}{4})\right].
    \label{epsilonn}
\end{align}
Since $\mbb E\left[ i_{\mbf g}(\bs{T}^n,\bs{Z}^n)\right]=n\mbb E\left[i_{\mbf g}(\bs{T},\bs{Z})\right], $ it follows from $\eqref{meanJ}$ using Lemma \ref{abweichungmean} that
\begin{align}
    \mbb P\left[i_{\mbf g}(\bs{T}^n,\bs{Z}^n)\leq n(R+\frac{\theta}{4})\right] 
    &=\mbb P\left[i_{\mbf g}(\bs{T}^n,\bs{Z}^n)\leq n(R+\frac{\theta}{2})-n\frac{\theta}{4}\right] \nonumber \\
    &\leq \mbb P\left[i_{\mbf g}(\bs{T}^n,\bs{Z}^n)\leq \mbb E\left[ i_{\mbf g}(\bs{T}^n,\bs{Z}^n)\right]-n\frac{\theta}{4}\right] \nonumber \\
    &\leq 2^{-\frac{nN_{R}}{2\ln(2)}\left[\left(1+\frac{(\ln(2)\theta)^2}{(4N_{R})^2}\right)^{\frac{1}{2}}-1\right]}. \nonumber
\end{align}
Then, it follows using \eqref{epsilonn} that for all $\mbf g \in \mc G'$
\begin{align}
      e(\Gamma_n,\mbf g) &\leq (\lvert \mathcal{G}^{\prime}\rvert+\lvert \mathcal{G}^{\prime} \rvert^2)2^{-\frac{n\theta}{8}}+\lvert \mathcal{G}^{\prime}\rvert 2^{-n\hat{\beta}}  +\lvert \mathcal{G}^{\prime} \rvert 2^{-\frac{nN_{R}}{2\ln(2)}\left[\left(1+\frac{(\ln(2)\theta)^2}{(4N_{R})^2}\right)^{\frac{1}{2}}-1\right]}.\nonumber
\end{align}
The last upper-bound is exponentially small for sufficiently large $n$. Since $R$ is any number less than $\underset{\mbf Q \in \mathcal{Q}_{P}}{\sup}\underset{\mbf g\in\mathcal{G}^{\prime}}{\inf} \log\det(\mathbf{I}_{N_{R}}+\frac{1}{\sigma^2}\mathbf{g}\mathbf{Q}\mathbf{g}^{H}),$  Theorem \ref{capacityfinitestate} is proved.
\end{proof} 
\subsection{Extension of the Proof for the set $\mc G$}
Now, we proceed with the proof of Theorem \ref{achievratecompoundchannels}.
In the proof, we will make use of  Theorem \ref{capacityfinitestate}. We will additionally use the following lemmas:

\begin{lemma}
\label{approximation}
Let $W_{\mbf g}$ and $W_{\hat{\mbf g}} $ be two channels such that $\mbf g,\hat{\mbf g} \in \mathcal{G}$ and let $\bs{t}^n$ be an input
$n$-sequence of vectors $\bs{t}_i$ such that $\frac{1}{n}\sum_{i=1}^{n}\lVert \bs{t}_i\rVert^2\leq P$ and let $\bs{z}^n$ be an output $n$-sequence of vectors $\bs{z}_i$ such that $\frac{1}{n} \sum_{i=1}^{n}\lVert \bs{z}_{i} \rVert^2\leq \rho, \  \rho>0.$
Then, it holds that 
\begin{align}
\frac{W_{\mbf g}(\bs{z}^n|\bs{t}^n)}{ W_{\hat{\mbf g}}(\bs{z}^n|\bs{t}^n)}\leq 2^{\frac{2n}{\ln(2)\sigma^2}\left[\sqrt{P\rho}+aP\right]\lVert \mbf{g}-\hat{\mbf g}\rVert}.\nonumber
\end{align}
\begin{proof}
$\forall i \in \{1,\hdots,n\},$ we have
\begin{align}
&\frac{W_{\mbf g}(\bs{z}_i|\bs{t}_i)}{ W_{\hat{\mbf g}}(\bs{z}_i|\bs{t}_i)} =\exp\left(-\frac{1}{\sigma^2}\left[(\bs{z}_i-\mbf g \bs{t}_i)^{H}(\bs{z}_i-\mbf g \bs{t}_i)-(\bs{z}_i- \hat{\mbf g} \bs{t}_i)^{H}(\bs{z}_i-\hat{\mbf g} \bs{t}_i)\right]\right), \nonumber
\end{align}
where
\begin{align}
   & -´\frac{1}{\sigma^2}\left[ \ \left(\bs{z}_i-\mbf g \bs{t}_i\right)^{H}\left(\bs{z}_i-\mbf g \bs{t}_i\right)-\left(\bs{z}_i- \hat{\mbf g} \bs{t}_i\right)^{H}\left(\bs{z}_i-\hat{\mbf g} \bs{t}_i\right) \right] \nonumber \\
    & \leq´\frac{1}{\sigma^2}\left| \left(\bs{z}_i-\mbf g \bs{t}_i\right)^{H}\left(\bs{z}_i-\mbf g \bs{t}_i\right)-\left(\bs{z}_i- \hat{\mbf g} \bs{t}_i\right)^{H}\left(\bs{z}_i-\hat{\mbf g} \bs{t}_i\right)\right| \nonumber \\
    &=\frac{1}{\sigma^2}\left| \bs{z}_i^{H}\left(\hat{\mbf g}-\mbf g\right)\bs{t}_i+\left(\bs{z}_i^{H}\left(\hat{\mbf g}-\mbf g\right)\bs{t}_i\right)^{H}      +\lVert \mbf g \bs{t}_i \rVert^2-\lVert \hat{\mbf g} \bs{t}_i \rVert^2 \right| \nonumber \\
    &\leq \frac{1}{\sigma^2}  \left|  \bs{z}_i^{H}\left(\hat{\mbf g}-\mbf g\right)\bs{t}_i+\left(\bs{z}_i^{H}\left(\hat{\mbf g}-\mbf g\right)\bs{t}_i\right)^{H} \right| +\frac{1}{\sigma^2} \left| \ \lVert \mbf g \bs{t}_i \rVert^2-\lVert \hat{\mbf g} \bs{t}_i \rVert^2 \right|             \nonumber\\
    &\leq \frac{1}{\sigma^2} \left[ 2\lVert \hat{\mbf g}-\mbf g \rVert \lVert \bs{t}_i\rVert  \lVert \bs{z}_i \rVert  +\left| \lVert \mbf g \bs{t}_i \rVert^2-\lVert \hat{\mbf g} \bs{t}_i \rVert^2  \right|   \right] \nonumber \\
    &=  \frac{2}{\sigma^2} \lVert \hat{\mbf g}-\mbf g \rVert \lVert \bs{t}_i\rVert  \lVert \bs{z}_i \rVert  +\frac{1}{\sigma^2} \left|\lVert \mbf g \bs{t}_i \rVert-\lVert \hat{\mbf g} \bs{t}_i \rVert  \right|\left(\lVert \mbf g \bs{t}_i \rVert+\lVert \hat{\mbf g} \bs{t}_i \rVert \right) \nonumber \\
    &\leq\frac{1}{\sigma^2}\left[ 2\lVert \hat{\mbf g}-\mbf g \rVert  \lVert \bs{t}_i\rVert  \lVert \bs{z}_i \rVert+\lVert\left(\mbf g - \mbf {\hat{g}}\right)\bs{t}_i\rVert  \left(\lVert \mbf g \bs{t}_i \rVert+\lVert \hat{\mbf g} \bs{t}_i \rVert\right)  \right]               \nonumber \\
    &\leq \frac{1}{\sigma^2}\left[ 2\lVert \hat{\mbf g}-\mbf g \rVert \lVert \bs{t}_i\rVert  \lVert \bs{z}_i \rVert+\lVert\left(\mbf g - \mbf {\hat{g}}\right)\bs{t}_i\rVert  \left(\lVert \mbf g \rVert \lVert \bs{t}_i \rVert+\lVert \hat{\mbf g} \rVert \lVert \bs{t}_i \rVert\right)  \right]               \nonumber \\
    & \leq  \frac{1}{\sigma^2} \left[  2\lVert \hat{\mbf g}-\mbf g \rVert  \lVert \bs{t}_i\rVert  \lVert \bs{z}_i \rVert +2 a \lVert \bs{t}_i \rVert  \lVert \hat{\mbf g} -\mbf g  \rVert  \lVert \bs{t}_i \rVert \right] \nonumber \\
    &=\frac{2}{\sigma^2}\lVert \hat{\mbf g} - \mbf g  \rVert \left[ \lVert \bs{t}_i\rVert \lVert \bs{z}_i \rVert + a\lVert \bs{t}_i \rVert^2   \right],\nonumber
\end{align}
where we used that $\lVert \mbf g \rVert \leq a$ and $\lVert \hat{\mbf g} \rVert \leq a$  for $\mbf g,\hat{\mbf g} \in \mathcal{G}\subset\mc G_a.$ 

Now
\begin{align}
    \frac{W_{\mbf g}(\bs{z}^n|\bs{t}^n)}{ W_{\hat{\mbf g}}(\bs{z}^n|\bs{t}^n)}&\overset{(a)}{=}\prod_{i=1}^{n}\frac{W_{\mbf g}(\bs{z}_i|\bs{t}_i)}{ W_{\hat{\mbf g}}(\bs{z}_i|\bs{t}_i)} \nonumber\\
    &\leq\exp\left(\sum_{i=1}^{n}\frac{2}{\sigma^2}\lVert \hat{\mbf g} - \mbf g  \rVert\left[  \lVert \bs{t}_i\rVert \lVert \bs{z}_i \rVert + a\lVert \bs{t}_i \rVert^2        \right]\right) \nonumber \\
    &=\exp\left(\frac{2}{\sigma^2}\lVert \hat{\mbf g} - \mbf g  \rVert\left[ \sum_{i=1}^{n} \lVert \bs{t}_i\rVert \lVert \bs{z}_i \rVert + a \sum_{i=1}^{n}\lVert \bs{t}_i \rVert^2        \right]\right) \nonumber \\
    &\overset{(b)}{\leq} \exp\left(\frac{2}{\sigma^2}\lVert \hat{\mbf g} - \mbf g  \rVert\left[ \sqrt{\sum_{i=1}^{n} \lVert \bs{t}_i\rVert^2} \sqrt{\sum_{i=1}^{n}\lVert \bs{z}_i \rVert^2} + a \sum_{i=1}^{n}\lVert \bs{t}_i \rVert^2        \right]          \right) \nonumber \\
    &\overset{(c)}{\leq} \exp\left(\frac{2n}{\sigma^2}\left[\sqrt{P\rho}+aP\right]\lVert \mbf{g}-\hat{\mbf g}\rVert\right),\nonumber
    & \nonumber \\&= 2^{\frac{2n}{\ln(2)\sigma^2}\left[\sqrt{P\rho}+aP\right]\lVert \mbf{g}-\hat{\mbf g}\rVert},
\nonumber \end{align}
where $(a)$ follows because the channels $W_\mbf g$ and $W_{\hat{\mbf g}}$ are memoryless, $(b)$ follows from Cauchy-Schwarz's inequality and $(c)$ follows because we require that $\frac{1}{n}\sum_{i=1}^{n}\lVert \bs{t}_i\rVert^2\leq P$ and $\frac{1}{n}\sum_{i=1}^{n}\lVert \bs{z}_i\rVert^2\leq \rho.$ This completes the proof of the lemma.
\end{proof}
\end{lemma}
\begin{lemma} 
\label{lemmadelta}
Let $\mbf g \in \mathcal{G}.$ Let $\bs{t}^n=(\bs{t}_1,\hdots,\bs{t}_n)$ be any $n$-input sequence of $W_\mbf g$ satisfying $\sum_{i=1}^{n} \lVert \bs{t}_i \rVert^2 \leq nP.$  Let $\bs{z}^n=(\bs{z}_1,\hdots,\bs{z}_n)$ be the $n$-output sequence. It holds that
\begin{align}
W_{\mbf g}\left( \sum_{i=1}^{n}\lVert \bs{z}_i\rVert^2 \geq n(2a^2P+2N_{R}\sigma^2+2)|\bs{t}^n\right) \leq\left[\left(1+\frac{1}{\sigma^2 N_{R}}\right) 2^{-\frac{1}{\ln(2)\sigma^2N_{R}}}\right]^{n}.\nonumber
\end{align}
\begin{proof}
We have
\begin{align}
\sum_{i=1}^{n} \lVert \bs{z}_i \rVert^2 
    &=\sum_{i=1}^{n} \lVert \mbf g \bs{t}_i+\bs{\xi}_i \rVert^2 \nonumber \\
    &\leq 2\sum_{i=1}^{n} \left( \lVert \bs{\xi}_i \rVert^2+\lVert \mbf g \bs{t}_i \rVert^2         \right) \nonumber \\
    &\leq 2\sum_{i=1}^{n} \left( \lVert \bs{\xi}_i \rVert^2+\lVert \mbf g \rVert^2 \lVert \bs{t}_i \rVert^2         \right) \nonumber \\
    &\leq 2 \sum_{i=1}^{n}\lVert \bs{\xi}_i \rVert^2 +2a^2 n P.\nonumber
\end{align}
Hence
\begin{align}
    W_{\mbf g}\left( \sum_{i=1}^{n} \lVert \bs{z}_i \rVert^2 \geq n(2a^2P+2N_{R}\sigma^2+2)|\bs{t}^n\right)&\leq \mbb P \left[ 2  \sum_{i=1}^{n} \lVert \bs{\xi}_i \rVert^2 +2a^2 n P \geq n(2a^2P+2N_{R}\sigma^2+2)\right]\nonumber \\
    &= \mbb P\left[ \sum_{i=1}^{n}\lVert \bs{\xi}_i \rVert^2 \geq n(N_{R}\sigma^2+1)\right]
    \nonumber \\
    &= \mbb P\left[ \sum_{i=1}^{n}\lVert \bs{\xi}_i \rVert^2 \geq n(\text{tr}(\sigma^2\mbf{I}_{N_{R}})+1)\right] \nonumber \\
    &\leq \left[\left(1+\frac{1}{\sigma^2 N_{R}}\right) 2^{-\frac{1}{\ln(2)\sigma^2N_{R}}}\right]^{n},
\nonumber \end{align}
where we used Lemma \ref{upperboundprobsum} in the last step. This completes the proof of the lemma.
\end{proof}
\end{lemma}
Now that we proved the lemmas, we fix $R$ to be any positive number less than $\underset{\mbf Q \in \mathcal{Q}_{P}}{\sup}\underset{\mbf g \in \mathcal{G}}{\inf} \log\det(\mbf I_{N_{R}}+\frac{1}{\sigma^2}\mbf g \mbf Q \mbf g^H)$ and put $2\theta=\underset{\mbf Q \in \mathcal{Q}_{P}}{\sup}\underset{\mbf g \in \mathcal{G}}{\inf} \log\det(\mbf I_{N_{R}}+\frac{1}{\sigma^2}\mbf g \mbf Q \mbf g^H)-R$.
By Lemma \ref{existence}, one can find a non-singular $\mbf Q_1 \in \mathcal{Q}_{P}$ such that $\text{tr}(\mbf Q_1)=P-\beta, \quad \beta>0$, and 
\begin{align}
 \mbb E\left[i_{\mbf g} \left(\bs{T},\bs{Z}\right)\right]=\log\det(\mbf I_{N_{R}}+\frac{1}{\sigma^2}\mbf g \mbf Q_{1} \mbf g) \geq R+\theta \quad \forall \mbf g \in \mathcal{G},
 \label{inequalityrate}
\end{align} 
with $\bs{T}$ and $\bs{Z}$ being the random input and output of $W_\mbf g,$ respectively.
We now pick a finite subset $\mathcal{G}^{\prime}$ of $\mathcal{G}$ such that for every $\mbf g \in \mathcal{G}$, there is a $\hat{\mbf g} \in \mathcal{G}^{\prime}$ satisfying $\lVert \mbf g - \hat{\mbf g}\rVert \leq \mu.$
This can be done because $\mathcal{G}$ is a bounded subset of a finite-dimensional Euclidean space and hence is totally bounded. By inequality \eqref{inequalityrate} and since 
\begin{align}
\underset{\mbf Q \in \mathcal{Q}_{P}}{\sup}\underset{\mbf g \in \mathcal{G}^{\prime}}{\inf} \log\det(\mbf I_{N_{R}}+\frac{1}{\sigma^2}\mbf g \mbf Q \mbf g^H) \geq \underset{\mbf Q \in \mathcal{Q}_{P}}{\sup}\underset{\mbf g \in \mathcal{G}}{\inf} \log\det(\mbf I_{N_{R}}+\frac{1}{\sigma^2}\mbf g \mbf Q \mbf g^H), \nonumber
\end{align}
it follows that
$$\underset{\mbf Q \in \mathcal{Q}_{P}}{\sup}\underset{\mbf g \in \mathcal{G}^{\prime}}{\inf} \log\det(\mbf I_{N_{R}}+\frac{1}{\sigma^2}\mbf g \mbf Q \mbf g^H)\geq R+\theta.
$$
Hence, the calculations of Theorem \ref{capacityfinitestate}  imply that there exists a code $\Gamma_n$ for $\mathcal{C}'$ with block-length $n$, size $\lvert \Gamma \rvert=\lfloor 2^{nR} \rfloor$ such that the codewords $\bs{t}^n=(\bs{t}_1,\hdots,\bs{t}_n)$ satisfy $\frac{1}{n}\sum_{i=1}^{n}\lVert \bs{t}_i \rVert^2\leq P, i=1,\hdots,n$ and such that for all $\hat{\mbf g} \in \mc G'$
    \begin{align}
     e(\Gamma_n,\hat{\mbf g}) & \leq (|\mathcal{G}^{\prime}|+|\mathcal{G}^{\prime}|^2)2^{-\frac{n\theta}{8}}+|\mathcal{G}^{\prime}|2^{-n\hat{\beta}} +\lvert \mathcal{G}^{\prime} \rvert 2^{-\frac{nN_{R}}{2\ln(2)}\left[\left(1+\frac{(\ln(2)\theta)^2}{(4N_{R})^2}\right)^{\frac{1}{2}}-1\right]},        \label{epsilonstrich}
    \end{align}
    where $\hat{\beta}=\frac{\beta}{\ln(2)(P-\beta)}-\log(1+\frac{\beta}{P-\beta})$ and where $\beta$ is independent of $n$.

We now consider the use of codewords and decoding sets belonging to the  code $\Gamma_n$ for $\mathcal{C}'$ with the larger compound channel $\mathcal{C}$.
Let $\mbf g \in \mathcal{G}$ and $\hat{\mbf g} \in \mathcal{G}^{\prime}$ such that $\lVert \mbf g -\hat{\mbf g} \rVert \leq \mu.$
Let $\bs{t}^n$ be \textit{any} codeword of $\Gamma_n$ and $B$ the corresponding decoding set. Let $F=\{\bs{z}^n=(\bs{z}_1,\hdots,\bs{z}_n):\frac{1}{n}\sum_{i=1}^{n}\lVert \bs{z}_i \rVert^2\leq \rho,i=1,\hdots,n\},$ where
$\rho=2a^2P+2N_{R}\sigma^2+2.$
Then 
 \begin{align}
     W_{\mbf g}(B^{c}|\bs{t}^n) &=W_{\mbf g}(B^{c}\cap F \cup B^c\cap F^c|\bs{t}^n) \nonumber \\
     &\leq W_{\mbf g}(B^{c}\cap F|\bs{t}^n)+W_{\mbf g}(F^c|\bs{t}^n).\nonumber
 \end{align}
 By Lemma \ref{lemmadelta}, it holds that
 \begin{align}
     W_{\mbf g}(F^{c}|\bs{t}^n)&\leq \left[ \left(1+\frac{1}{N_{R}\sigma^2}\right)2^{-\frac{1}{\ln(2)N_{R}\sigma^2}}   \right]^{n} \nonumber \\
     &=2^{-n \left( \frac{1}{\ln(2)N_{R}\sigma^2}-\log\left(1+\frac{1}{N_{R}\sigma^2} \right)\right)}.\nonumber
 \end{align}
By Lemma \ref{approximation}, it holds that
\begin{align}
    &W_{\mbf g}(B^{c}\cap F|\bs{t}^n) \leq 2^{\frac{2n}{\ln(2)\sigma^2}\left[\sqrt{P\rho}+aP\right] \mu }  W_{\hat{\mbf g}}(B^{c}\cap F|\bs{t}^n). \nonumber
\end{align}
Now
\begin{align}
    W_{\hat{\mbf g}}(B^c\cap F|\bs{t}^n)\leq W_{\hat{\mbf g}}(B^c|\bs{t}^n)\leq e(\Gamma_n,\hat{\mbf g}).\nonumber
\end{align}
This implies using \eqref{epsilonstrich} that for all $\mbf g \in \mc G$ 
\begin{align}
 \ e(\Gamma_n,\mbf g) &\leq 2^{-n \left( \frac{1}{\ln(2)N_{R}\sigma^2}-\log\left(1+\frac{1}{N_{R}\sigma^2} \right)\right)}+ (|\mathcal{G}^{\prime}|+|\mathcal{G}^{\prime}|^2) 
    2^{-n\left(\frac{\theta}{8}-\frac{2}{\ln(2)\sigma^2}\left[\sqrt{P\rho}+aP\right] \mu\right)}\nonumber \\
    &\quad+|\mathcal{G}^{\prime}|2^{-n\left(\hat{\beta}-\frac{2}{\ln(2)\sigma^2}\left[\sqrt{P\rho}+aP\right] \mu \right)}  +|\mathcal{G}^{\prime}|2^{-n\left[ c_1 -c_{2}\mu       \right]},
    \label{exponentials}
\end{align}
where
\begin{align}
    c_1=\frac{N_{R}}{2\ln(2)}\left[\left(1+\frac{(\ln(2)\theta)^2}{(4N_{R})^2}\right)^{\frac{1}{2}}-1\right] \nonumber
\end{align}
and
\begin{align}
   c_2= \frac{2}{\ln(2)\sigma^2}\left[\sqrt{P\rho}+aP\right].\nonumber 
\end{align}
The exponentials in $\eqref{exponentials}$ are all of the form $2^{-n(K_{1}-K_2\mu)}$ where $K_1$ and $K_2$ do not depend on $n$ and where $K_1$ is positive and $K_2$ is non-negative.
Consequently for $\mu$ sufficiently small, it follows that 
$\underset{n\rightarrow\infty}{\lim} e(\Gamma_n,\mbf g)=0. $
This proves that $\underset{\mbf Q \in \mathcal{Q}_{P}}{\sup}\underset{\mbf g \in \mathcal{G}}{\inf} \log\det(\mbf I_{N_R}+\frac{1}{\sigma^2}\mbf g \mbf Q \mbf g^H)$ is an achievable rate for $\mathcal{C}.$ This completes the direct proof of Theorem \ref{achievratecompoundchannels}.
\section{Application of Common Randomness Generation: Correlation-assisted identification}
\label{application}
In this section, we study the problem of correlation-assisted identification over MIMO slow fading channels, as an application of CR generation.
\subsection{Definitions}
We provide the definition of the $\eta$-outage correlation-assisted identification capacity. For this purpose, we start by defining a correlation-assisted identification code for the MIMO slow fading channel $W_\mbf G.$ In what follows, $x^n$ and $y^n$ are any realizations of $X^n$ and $Y^n$, respectively.

\begin{definition}
A correlation-assisted identification-code  of length $n$ and size $N$ for the MIMO slow fading channel $W_{\mbf G}$  is a family of pairs of codewords and decoding regions 
 $\left\{(\mbf{t}_\ell(x^n),\setd^{(\mbf g)}_\ell(y^n)),\mbf g\in \mbb C^{N_R\times N_T}, \ell=1,\ldots, N  \right\}$ such that for all $\ell \in \{1,\ldots, N \}$ and all $\mbf g \in \mbb C^{N_R´\times N_T}$, we have
\begin{align}
& \mbf{t}_\ell(x^n) \in \mbb C^{N_{T}\times n},\quad \setd^{(\mbf g)}_\ell(y^n) \subset \mbb C^{N_{R}\times n}, \nonumber \\
&\frac{1}{n}\sum_{i=1}^{n}\bs{t}_{\ell,i}^H\bs{t}_{\ell,i}\leq P \ \ \mbf{t}_\ell(x^n)=(\bs{t}_{\ell,1},\hdots,\bs{t}_{\ell,n}). \nonumber 
\end{align}
The error of first and second kind are expressed as
\begin{align}
    E_{1}(\mbf g,y^n)=  \underset{\ell \in \{1,\ldots,N\}}{\max}W_{\mbf G}(\setd_\ell^{(\mbf g)}(y^n)^c|\mbf{t}_\ell(x^n)).\nonumber
\end{align}
and
\begin{align}
    E_{2}(\mbf g,y^n)=
    \underset{\ell \in \{1,\ldots,N\},\ell\neq j}{\max}W_{\mbf G}(\setd_\ell^{(\mbf g)}(y^n)|\mbf{t}_j(x^n)). \nonumber
\end{align}
\end{definition}

\begin{definition}
    \label{defcapacity}
  $C_{\eta,ID}^{c}(P,N_T\times N_R).$  the $\eta$-outage correlation-assisted  identification capacity of the MIMO slow fading channel $W_\mbf G,$ is defined as follows:
  \begin{align}
      C_{\eta,ID}^{c}(P,N_T\times N_R)=\sup\left\{ R\colon \forall \lambda>0,\ \exists n(\lambda) \text{ s.t. for } n \geq n(\lambda)  \ N(n,\lambda) \geq 2^{2^{nR}}         \right\}, \nonumber
  \end{align}
 where $N(n,\lambda)$ is the maximal cardinality such that a  correlation-assisted identification code with length $n$ for the channel $W_\mbf G$ exists such that for some $\lambda_1,\lambda_2\leq\lambda$ with $\lambda_1+\lambda_2<1$, the following is satisfied
   \begin{align}
      \mbb P \left[ E_1(\mbf G,Y^n)\leq \lambda_1\right]\geq 1-\eta \nonumber
  \end{align}
  and
   \begin{align}
     \mbb P \left[ E_2(\mbf G,Y^n)\leq \lambda_2\right]\geq 1-\eta. \nonumber
  \end{align}
\end{definition}

\subsection{Lower bound on the outage correlation-assisted identification capacity}
In this section, we will proceed analogously to \cite{identificationcode} we will establish a lower bound on the $\eta$-outage correlation-assisted identification capacity, provided in the following theorem.
\begin{theorem}
\begin{align}
    C_{\eta,Id}^{c}(P,N_{T}\times N_{R}
)\geq   \underset{ \substack{U \\{\substack{U \circlearrow{X} \circlearrow{Y}\\ I(U;X)-I(U;Y) \leq C_{\eta}(P,N_{T}\times N_{R})}}}}{\max} I(U;X).  \nonumber
\end{align}
\end{theorem}
\begin{proof}
 Given a discrete memoryless multiple source $P_{XY}$, Alice observes the output $X^{n}$ and Bob observes the output $Y^{n}.$
 Alice generates a random variable $K$ with alphabet $\mathcal{K}=\{1,\dots, M^{\prime}\},$ such that $K=\Phi(X^{n})$ and $M^{\prime} \leq 2^{cn}$, with $c$ being a constant not depending on $n$. To send a message $i$, we prepare a set of coloring-functions or mappings $E_i$ known by the sender and the receiver.\begin{align*}
 E_i  & \colon \mathcal{K} \longrightarrow \{1,\ldots,M^{\prime\prime} \} \\
 & \colon \underbrace{K}_{\text{coloring}} \mapsto \underbrace{E_i(K).}_{\text{color}}
 \end{align*}
 $X^{n}$ is encoded to a sequence $\bs{T}^{n}$ satisfying 
 \begin{equation}
    \frac{1}{n}\sum_{i=1}^{n}\bs{T}_{i}^{H}\bs{T}_i\leq P, \quad \text{almost surely}.  \nonumber
\end{equation}
This is done by using a code $\Gamma_n$ such that
     \[
        \frac{\log\lvert \Gamma_n\rvert}{n}= C_{\eta}(P,N_{R}\times N_{T})-\delta
    \]
    and
    \[
        \mbb P[e(\Gamma_n,\mbf G)\leq\theta]\geq 1-\eta.
    \]
    
  $E_i(K)$ is encoded to a sequence $\bs{T}^{\lceil \sqrt{n} \rceil}$ satisfying:
  \begin{equation}
    \frac{1}{\lceil \sqrt{n} \rceil}\sum_{i=1}^{\lceil \sqrt{n} \rceil}\bs{T}_{i}^{H}\bs{T}_i\leq P. \nonumber
\end{equation}
This is done by using a code $\Gamma_{ \lceil \sqrt{n} \rceil}$ with vanishing rate equal to $\delta$ and where    \[
        \mbb P[e(\Gamma_{ \lceil \sqrt{n} \rceil},\mbf G)\leq\alpha]\geq 1-\eta.
    \]

By concatenating both sequences, we obtain the sequence $\bs{T}^{m}$, where $m=n+\lceil \sqrt{n} \rceil.$ $\bs{T}^m$ is sent over the MIMO slow fading channel.
  Bob generates $L=\Psi(Y^{n},\bs{Z}^{m})$ such that $\mbb P\left[\mbb P[K\neq L|\mbf G] \leq \alpha \right]\geq 1-\eta.$
  Bob identifies whether the message of interest was sent or not such that
  \begin{align}
      \mbb P \left[ E_1(\mbf G,Y^n)\leq \lambda_1\right]\geq 1-\eta 
      \label{outageerror1}
  \end{align}
   \begin{align}
     \mbb P \left[ E_2(\mbf G,Y^n)\leq \lambda_2\right]\geq 1-\eta.
     \label{outageerro2}
  \end{align}
for some $\lambda_1,\lambda_2\leq\lambda$ with $\lambda_1+\lambda_2<1.$
  We choose the rate of the first code to be approximately equal to  the $\eta$-outage capacity of the MIMO slow fading channel  so that Bob can identify the message,  with error of first and second kind satisfying \eqref{outageerror1} and \eqref{outageerro2}, respectively,  and at a rate approximately equal to $C_{\eta}(P,N_{T}\times N_{R}),$ the outage transmission capacity  of the  channel, without paying a price for the identification task.
By extending the Transformator lemma\cite{Generaltheory} to the outage setting, it follows that the $\eta$-outage correlation-assisted identification capacity of the MIMO slow fading is lower-bounded by its corresponding $\eta$-outage CR capacity. 
This implies using Theorem \ref{ccretathmMIMO} that
\begin{align}
    C_{\eta,Id}^{c}(P,N_{T}\times N_{R}
)\geq   \underset{ \substack{U \\{\substack{U \circlearrow{X} \circlearrow{Y}\\ I(U;X)-I(U;Y) \leq C_{\eta}(P,N_{T}\times N_{R})}}}}{\max} I(U;X).  \nonumber
\end{align}
\end{proof}
\section{Conclusion}
\label{conclusion}
In this paper, we introduced the concept of capacity versus outage in the CR generation framework to assess the performance in point-to-point MIMO slow fading environments with AWGN and with arbitrary state distribution.We established a single-letter characterization of the outage
CR capacity of the MIMO slow fading channel with AWGN
and with arbitrary state distribution using our result on its
 outage transmission capacity. The obtained results are particularly useful in the problem of correlation-assisted identification over MIMO slow fading channels.
As a future work, it would be interesting to study the problem of CR generation in fast fading environments, where the channel state varies over the time scale of transmission.

\vspace{12pt}


\begin{thebibliography}{00}
\bibitem{survey}M. Sudan, H. Tyagi and S. Watanabe, "Communication for Generating Correlation: A Unifying Survey," in IEEE Transactions on Information Theory, vol. 66, no. 1, pp. 5-37, Jan. 2020.
\bibitem{capacityAVC}I. Csiszár and P. Narayan, "The capacity of the arbitrarily varying channel revisited: positivity, constraints," in IEEE Transactions on Information Theory, vol. 34, no. 2, pp. 181-193, March 1988.
\bibitem{commitmentcapacity}Winter, A. et al. “Commitment Capacity of Discrete Memoryless Channels.” arXiv cs.CR/0304014, 2003.
\bibitem{unconditionallysecure} Rivest, R.L.: Unconditionally secure commitment and oblivious transfer schemes using private channels and a trusted
initializer, 1999.
\bibitem{identification}R. Ahlswede and G. Dueck, "Identification via channels," in IEEE Transactions on Information Theory, vol. 35, no. 1, pp. 15-29, Jan. 1989.
\bibitem{Generaltheory}R. Ahlswede, General theory of information transfer: updated, Discrete Applied Mathematics, Vol. 156, No. 9, 1348-1388, 2008.
\bibitem{part2}R. Ahlswede and I. Csiszár, "Common randomness in information theory and cryptography. II. CR capacity," in IEEE Transactions on Information Theory, vol. 44, no. 1, pp. 225-240, Jan. 1998.
\bibitem{CRincrease}A. Ahlswede, I. Althöfer, C. Deppe, and T. Ulrich, Identification and
Other Probabilistic Models Rudolf Ahlswede’s Lectures on Information
Theory 6, 1st ed. Springer-Verlag, 2021, vol. 16.
\bibitem{shannon}C.  E.  Shannon,  “A  mathematical  theory  of  communication,”Bell  System  Technical  Journal,  vol.  27,  pp.  379–423,  623–656,  July, October 1948.
\bibitem{applications}H. Boche and C. Deppe, "Secure Identification for Wiretap Channels; Robustness, Super-Additivity and Continuity," in IEEE Transactions on Information Forensics and Security, vol. 13, no. 7, pp. 1641-1655, July 2018.
\bibitem{tactiles}G. P. Fettweis, "The Tactile Internet: Applications and Challenges," in IEEE Vehicular Technology Magazine, vol. 9, no. 1, pp. 64-70, March 2014.
\bibitem{Moulin}P. Moulin, “The role of information theory in watermarking and
its application to image watermarking,” Signal Processing, vol. 81,
no. 6, pp. 1121 – 1139, 2001, special section on Information
theoretic aspects of digital watermarking. 
\bibitem{watermarkingahlswede} R.  Ahlswede  and  N.  Cai, Watermarking  Identification  Codes  with  Related  Topics  on  Common  Randomness. Berlin,  Heidelberg:Springer Berlin Heidelberg, 2006, pp. 107–153.
\bibitem{watermarking}Y. Steinberg and N. Merhav, "Identification in the presence of side information with application to watermarking," in IEEE Transactions on Information Theory, vol. 47, no. 4, pp. 1410-1422, May 2001.
\bibitem{industrie40}Y. Lu, “Industry 4.0: A survey on technologies, applications and open
research issues,” Journal of Industrial Information Integration, vol. 6,
pp. 1 – 10, 2017.
\bibitem{implementation}Derebeyoğlu, Sencer; Deppe, Christian; Ferrara, Roberto: Performance Analysis of Identification Codes. Entropy 22 (10), Sep 2020.
\bibitem{PatentBA}H. Boche and C. Arendt, “Communication method, mobile unit, interfaceunit, and communication system,” 2021, patent number: 10959088.
\bibitem{part1}R. Ahlswede and I. Csiszár, "Common randomness in information theory and cryptography. I. Secret sharing," in IEEE Transactions on Information Theory, vol. 39, no. 4, pp. 1121-1132, July 1993.
\bibitem{Maurer}U. M. Maurer, "Secret key agreement by public discussion from common information," in IEEE Transactions on Information Theory, vol. 39, no. 3, pp. 733-742, May 1993.
\bibitem{CRgaussian}R. Ezzine, W. Labidi, H. Boche and C. Deppe, "Common Randomness Generation and Identification over Gaussian Channels," GLOBECOM 2020 - 2020 IEEE Global Communications Conference, 2020, pp. 1-6.
\bibitem{wafapaper}W. Labidi, C. Deppe and H. Boche, "Secure Identification for Gaussian Channels," ICASSP 2020 - 2020 IEEE International Conference on Acoustics, Speech and Signal Processing (ICASSP), 2020, pp. 2872-2876.
\bibitem{Tse}D. Tse and P. Viswanath, Fundamentals of Wireless Communication, 2005.
\bibitem{goldsmith}A. Goldsmith, Wireless Communications. Cambridge University Press, 
2005.
\bibitem{inftheoretic}L. H. Ozarow, S. Shamai and A. D. Wyner, "Information theoretic considerations for cellular mobile radio," in IEEE Transactions on Vehicular Technology, vol. 43, no. 2, pp. 359-378, May 1994.
New York, NY, USA: Cambridge University Press.
\bibitem{infaspects}E. Biglieri, J. Proakis and S. Shamai, "Fading channels: information-theoretic and communications aspects," in IEEE Transactions on Information Theory, vol. 44, no. 6, pp. 2619-2692, Oct. 1998.
\bibitem{telatar}Telatar, E. “Capacity of Multi-antenna Gaussian Channels.” Eur. Trans. Telecommun. 10 (1999): 585-595.
\bibitem{codingtheorems}I. Csiszár and J. Körner, Information Theory: Coding Theorems for
Discrete Memoryless Systems, 2nd ed. Cambridge University Press,
2011.
\bibitem{NoteShannon}N. J. A. Sloane; Aaron D. Wyner, "A Note on a Partial Ordering for Communication Channels," in Claude E. Shannon: Collected Papers , IEEE, 1993, pp.265-272.
\bibitem{discretetimegaussian}Root, W. L., and P. P. Varaiya. “Capacity of Classes of Gaussian Channels.” SIAM Journal on Applied Mathematics, vol. 16, no. 6, 1968, pp. 1350–1393.
\bibitem{errorbound}A. Thomasian, “Error bounds for continuous channels,” inProc. 4th LondonSymp. Inf. Theory.  Washington, DC, 1961, pp. 46–60.
\bibitem{informationbook} R. Ash, Information Theory, ser. Interscience tracts in pure and applied mathematics.  Interscience Publishers, 1965.
\bibitem{capacityofclassofchannels}David Blackwell, Leo Breiman, A. J. Thomasian "The Capacity of a Class of Channels," The Annals of Mathematical Statistics, Ann. Math. Statist. 30(4), 1229-1241, (December, 1959).
\bibitem{identificationcode}R. Ahlswede and G. Dueck, "Identification in the presence of feedback-a discovery of new capacity formulas," in IEEE Transactions on Information Theory, vol. 35, no. 1, pp. 30-36, Jan. 1989.
\end{thebibliography}
\end{document}